\documentclass[11pt, letterpaper]{article}

\pdfoutput=1

\usepackage{chngcntr}
\usepackage{url}
\usepackage{dsfont}
\usepackage[final]{graphicx}
\usepackage{amsmath}
\usepackage{amssymb,amsthm}
\usepackage{algorithm}
\usepackage{algorithmic}
\usepackage{appendix}
\usepackage{hyperref}
\usepackage{bm}
\usepackage{cite}
\usepackage{enumerate}
\usepackage[papersize={8.5in,11in},top=1in,left=1in,right=1in,bottom=1in]{geometry}
\usepackage{xfrac}
\usepackage{sidecap}

\makeatletter
\def\@seccntformat#1{\@ifundefined{#1@cntformat}%
   {\csname the#1\endcsname\quad}  
   {\csname #1@cntformat\endcsname}
}
\let\oldappendix\appendix 
\renewcommand\appendix{%
    \oldappendix
    \newcommand{\section@cntformat}{\appendixname~\thesection\quad}
}
\makeatother

\newenvironment{lessspaceitemize*}%
  {\begin{itemize}%
  \vspace{-1mm}
    \setlength{\itemsep}{0pt}%
    \setlength{\parskip}{0pt}}%
    {\vspace{-1mm} \end{itemize}}

\newenvironment{lessspaceenum*}%
  {\begin{enumerate}%
  \vspace{-1mm}
    \setlength{\itemsep}{0pt}%
    \setlength{\parskip}{0pt}}%
  {\end{enumerate}}

\counterwithout{figure}{section}

\newenvironment{definition}[1][Definition]{\begin{trivlist}
\item[\hskip \labelsep {\bfseries #1}]}{\end{trivlist}}
\newenvironment{assumption}[1][Assumption]{\begin{trivlist}
\item[\hskip \labelsep {\bfseries #1}]}{\end{trivlist}}

\DeclareMathOperator*{\argmin}{arg\,min}

\newtheorem{theorem}{Theorem}[section]
\newtheorem{lemma}[theorem]{Lemma}
\newtheorem{claim}[theorem]{Claim}
\newtheorem{proposition}[theorem]{Proposition}
\newtheorem{corollary}[theorem]{Corollary}

\newenvironment{thm_app}[1]{\noindent\textbf{Theorem~\ref{#1}.}}{\par\addvspace{\baselineskip}}
\newenvironment{lem_app}[1]{\noindent\textbf{Lemma~\ref{#1}.}}{\par\addvspace{\baselineskip}}

\newenvironment{prop_app}[1]{\noindent\textbf{Proposition~\ref{#1}.}}{\par\addvspace{\baselineskip}}

\newcommand{\lx}{\ensuremath{\lambda(x)}}%
\newcommand{\ldx}{\ensuremath{\lambda'(x)}}%
\newcommand{\bi}{\ensuremath{\overline{p}_i}}%

\title{Envy-Free Pricing in Large Markets: Approximating Revenue and Welfare}

\author{Elliot Anshelevich\thanks{Computer Science Dept, RPI, Troy, NY. \texttt{eanshel@cs.rpi.edu}}
\and Koushik Kar\thanks{Dept. of Electrical, Computer \& Systems Engg., RPI, Troy, NY \texttt{koushik@ecse.rpi.edu}}
\and Shreyas Sekar\thanks{Computer Science Dept, RPI, Troy, NY. \texttt{sekars@rpi.edu}}}
\begin{document}


\maketitle

\begin{abstract}
We study the classic setting of envy-free pricing, in which a single seller chooses prices for its many items, with the goal of maximizing revenue once the items are allocated. Despite the large body of work addressing such settings, most versions of this problem have resisted good approximation factors for maximizing revenue; this is true even for the classic unit-demand case. In this paper we study envy-free pricing with unit-demand buyers, but unlike previous work we focus on {\em large} markets: ones in which the demand of each buyer is infinitesimally small compared to the size of the overall market. We assume that the buyer valuations for the items they desire have a nice (although reasonable) structure, i.e., that the aggregate buyer demand has a monotone hazard rate and that the values of every buyer type come from the same support.

For such large markets, our main contribution is a $1.88$ approximation algorithm for maximizing revenue, showing that good pricing schemes can be computed when the number of buyers is large. We also give a $(e,2)$-bicriteria algorithm that simultaneously approximates both maximum revenue and welfare, thus showing that it is possible to obtain both good revenue and welfare at the same time. We further generalize our results by relaxing some of our assumptions, and quantify the necessary tradeoffs between revenue and welfare in our setting. Our results are the first known approximations for large markets, and crucially rely on new lower bounds which we prove for the revenue-maximizing prices.
\end{abstract}

\setcounter{page}{1}
\section{Introduction}
How should a seller controlling multiple goods choose prices for these goods, so that these prices yield good revenue and yet are efficiently computable? This question is among the most fundamental of algorithmic challenges motivated by Economic paradigms. At a high level, this setting can be modeled as a two-stage game: the seller chooses prices, and the buyers respond by purchasing goods at these prices. A common constraint in this context is one of \emph{envy-freeness}, i.e., every buyer receives items that maximize her utility, and thus would not want to ``switch places" with any other buyer.

Despite the surge of papers studying envy-free pricing in recent years~\cite{briest2011buying, cheung2008approximation, guruswami2005profit}, even the simplest versions of this problem have resisted good approximation factors for maximizing revenue. This is true even for the common setting of {\em unit-demand} buyers, where every buyer desires one unit of good from a demand set $S_i$ (possibly different for each buyer $i$); she values all items in $S_i$ equally and has no value for items outside $S_i$. The problem of revenue-maximization with unit-demand buyers is among the most popular versions of the pricing problem. While the best known approximation algorithm for the item-pricing version of this problem has only a logarithmic factor~\cite{briest2011buying, guruswami2005profit}, more sophisticated pricing mechanisms have yielded some beautiful, near-optimal mechanisms, but only by giving up envy-freeness~\cite{feldman2015posted, chawla2010multi}.

In this paper, we study envy-free pricing with unit-demand buyers, and form good approximation algorithms for maximizing both revenue and welfare. Unlike most previous work on this subject, we focus on {\em large} markets: ones in which the demand of each buyer is infinitesimally small compared to the market size. For envy-free settings, studying large markets is much more reasonable than markets with only a few buyers. Indeed, in such a market, a seller may not be able to price discriminate (i.e., sell the same good to at different prices), and would instead simply post a price for each good, which would apply to all of the buyers. The fact that all buyers who receive a copy of the same good pay the same price, along with buyers always purchasing a unit of the {\em cheapest} good in their set $S_i$, would guarantee that the allocation is envy-free.

\subsubsection*{Our Model}
We consider a single monopolist producing a set $S$ of goods, which are near-substitutes (see examples below). The seller can produce any desired quantity $x_t$ of a good $t \in S$, for which he incurs a cost of $C_t(x_t)$. The seller's main objective is to set prices on the goods to maximize revenue; in addition to revenue, the seller may also be interested in welfare guarantees. The market consists of a set $B$ of buyer types: for a given type $i \in B$, all the buyers having this type desire the same set $S_i \subseteq S$ of items. Every individual buyer's demand is infinitesimal compared to the market size. Therefore, we can represent every type $i \in B$ by a (inverse) demand function $\lambda_i(x)$ such that for any given $v$, we know how many buyers $x$ have a valuation of $v$ or more for items in $S_i$.

Although different buyer types may have different demand functions, it is natural to assume that the valuations of all buyers are often sampled (albeit differently) from some global distribution. Because of this, we will make the assumption that the buyer valuations for every type have the same support $[\lambda^{min}, \lambda^{max}]$ (although we will relax this assumption later). We show that as long as this is true and that the distribution of buyer valuations has reasonable structure (i.e., monotone hazard rate \cite{bagnoli2005log}), then we can compute prices which extract more than half of the optimal revenue. Our model captures several scenarios of interest; we illustrate two of them below.


\begin{enumerate}
\item \textbf{PEV Charging}: As Plug-in Electric Vehicles become commonplace, it is expected that charging stations will be set up at many locations. Due to the variable cost of electricity generation, these stations may have different prices for charging during different time intervals. We can model each time slot as an item $t$; every buyer has a set of time slots during which she can charge, and the seller may be able to predict the demand using prior data~\cite{tushar2012economics}.

\item \textbf{Display Advertising}: A publisher may have a set of items (e.g., advertising slots) being sold via simultaneous posted price auctions. The ad-items are differentiated (in their position or location on the website) and a large number of buyers are interested in buying these items, each interested in some specific subset depending on their target audience.
\end{enumerate}

Our model retains the combinatorial flavor of the general envy-free pricing problem: different buyer types have access to different subsets of items, and these subsets are not correlated in any way. It is this combinatorial aspect which contributes to the hardness of the problem. In fact, recent complexity results~\cite{briest2011buying, chalermsook2012improved} indicate that the general unit-demand problem with uniform valuations may not admit approximation algorithms with factors better than $O(\log|B|)$. The starting point of our work is the fact that in large markets with many buyer types, $O(\log|B|)$ algorithms are not acceptable. A large body of work has circumvented this hardness by studying interesting instances in which the combinatorial aspect of the model is limited or rendered moot~\cite{chen2014envy, chawla2007algorithmic, feldman2012revenue}. In contrast, we impose no such restriction on the model, instead making the assumption that the buyer valuations follow a nice (monotone hazard rate) structure, while the sets $S_i$ can be arbitrary. For the large market settings we are interested in, our assumptions seem more reasonable than restricting demand sets.

Two aspects of large markets that we will feature in this paper warrant further discussion. First, while the majority of literature has focused on envy-free pricing to maximize revenue (see Related Work for exceptions), we focus on maximizing both revenue and social welfare, and the trade-offs therein. This is motivated by the fact that in large markets with repeated engagement, compromising on welfare may often lead to poor revenue in the long-run. Second, in our model sellers face convex production costs $C_t(x)$ for each item $t$. This strictly generalizes models with limited or unlimited supply which are usually the norm. In large markets, assuming limited supply is too rigid as sellers may often be able to increase production, albeit at a higher cost. Costs, however, are a non-trivial addition to the envy-free model. Many of the standard techniques that previously yielded good algorithms, especially single-pricing for all items, fail to do so in our framework. The seller now faces the onerous task of balancing demand with production costs, which may be different for different items.

\subsection{Our Results}
Recall that every buyer type $i \in B$ is represented by an inverse demand function $\lambda_i(x)$ such that for any $v$, $\lambda_i(x) = v$ indicates that $x$ population of buyers hold a value of $v$ or more for the items in the demand set. Throughout this work, we will assume that $\forall i$, $\lambda_i(x)$ has a Monotone Hazard Rate (MHR, see Section 2 for definition). This coincides with the inverse demand being log-concave ($\log(\lambda_i(x))$ is concave) and encompasses several popular inverse demand functions previously considered in literature, including concave, power-law, and exponential demand~\cite{bagnoli2005log}. 

Our main contribution in this work is a {\em $1.88$ approximation algorithm for maximizing revenue} and a {\em $(e,2)$-bicriteria algorithm that simultaneously approximates both maximum revenue and welfare} respectively. These results hold as long as the `peak of the support' of the buyer demand function is the same for all buyer types, i.e., $\forall i$ $\lambda_i(0) = \lambda^{max}$. Notice that the setting where all buyer types have the same support $[\lambda^{min}, \lambda^{max}]$ is a special case of our uniform peak assumption. Although both of our results use a continuous ascending-price algorithm, we describe an efficient implementation for this algorithm as well. Note that revenue-maximization is still NP-Hard in this setting, due to its combinatorial nature.

We next generalize the uniform peak (or support) assumption and consider markets where every buyer type $i$ has a (potentially different) support $[\lambda^{min}_i, \lambda^{max}_i]$. For this setting, our results are parameterized by a factor $\Delta$ that equals the ratio of the maximum $\lambda^{max}_i$ to the minimum $\lambda^{max}_i$ across buyer types. We show a $O(\log\Delta)$ approximation to the optimal revenue in this setting, and thus imply that as long as the valuations for different buyer types are not too different, we can still extract high revenue. Moreover, we show that this $O(\log\Delta)$ solution also guarantees one fourth of the optimum social welfare. Although the actual buyer demand may be quite asymmetric, our result depends only on the difference in the peak of the supports; it is reasonable to expect that this difference is not too large if the goods are similar.

We now summarize the two high-level contributions that enable our results.
\begin{enumerate}
\item We provide a general framework to derive good algorithms for large markets with {\em production costs}, extensively using techniques from the theory of min-cost flows.

\item Our constant-approximation factors depend crucially on the insight that we gain on the prices in the revenue-maximizing solution. In contrast to previous work, where the approximation factor of the revenue of the computed solution is usually obtained by comparing it to the optimum social welfare~\cite{cheung2008approximation, guruswami2005profit} (which is an upper bound on optimum revenue), we are able to directly compare the revenue of our solution to the profit-maximizing solution.	
\end{enumerate}

\subsection{Related Work}
Our work is a part of a rather extensive body of literature studying envy-free or item-pricing; the field is too vast to survey here and we will only sample the most relevant results. The Unit-Demand Pricing (UDP) problem where buyers have different valuations for different items was first considered in~\cite{guruswami2005profit}, which gave a $O(\log|B|)$ approximation algorithm for maximizing revenue. The version that we study (each buyer has equal valuation for all items in $S_i$, and $0$ otherwise) has been referred to as UDP-MIN or UDP with Uniform Valuations. Surprisingly, the addition of uniform values has not lead to any improved algorithms for the general UDP problem. Moreover, recent complexity results~\cite{briest2011buying, chalermsook2012improved} indicate that a sub-logarithmic approximation factor may be unlikely for both problems.

Assuming more structure on the combinatorial aspect of UDP (i.e., sets $S_i$ stating which buyers have access to which items) has yielded more tractable instances. For example, good approximation algorithms exist when each item is desired by at most $k$ buyer types~\cite{chen2014envy, ImLW12}. For settings with budgeted buyers who have access to all items but have a limit on the amount of money then can spend, \cite{feldman2012revenue} give a $0.5$-approximation algorithm; we remark that budgeted buyers can be captured with an inverse demand $\lambda(x) = \sfrac{c}{x}$. In contrast, ours is among the few papers that makes no assumptions on the demand sets $S_i$ but still obtains a constant approximation factor. Finally, another active line of work has looked at envy-free pricing when each buyer demands a single bundle of items (Single-Minded Pricing). For more details, the reader is asked to refer to~\cite{briest2011buying},~\cite{cheung2008approximation}, and some of the references therein.

More broadly, our work bears certain similarities to algorithmic pricing mechanisms~\cite{chawla2007algorithmic} in a Bayesian setting, especially posted price mechanisms. In fact, the aggregate demand that we consider can be interpreted as buyers deriving values from a known distribution. Although posted pricing provides excellent guarantees, even in multi-parameter settings~\cite{chawla2010multi,feldman2015posted}, the mechanisms seldom result in envy-free allocations because it is assumed that buyers choose items in some order. At a high level, our work is a part of the literature exploring the space of multi-parameter settings with some structure. In addition to a valuation, buyers have a demand set ($S_i$) in our model, whereas researchers have looked at other models where the additional parameter is the quantity demanded~\cite{chen2014revenue} or a position in a metric space~\cite{chen2011metric}.

Finally, envy-free pricing to maximize welfare coincides with the notion of Walrasian Equilibrium minus the market clearing constraint. In large markets such as ours, Walrasian Equilibria are guaranteed to exist~\cite{azevedo2013walrasian}, although their revenue may be poor. In discrete markets, Walrasian prices are not guaranteed to exist and so, the focus has been on solutions that are \emph{approximately envy-free} but still guarantee good welfare~\cite{feldman2013combinatorial,chen2008walrasian}. There has also been some work on approximating both revenue and welfare over a restricted space of solutions; for instance, the space of all equilibria in GSP~\cite{lucier2012revenue}, or all competitive equilibria for \emph{sharp} multi-unit demand~\cite{chen2014revenue}. In contrast, bi-criteria approximations like ours, which compare both objectives for the same solution to the unrestricted global optima, have not been previously considered in envy-free literature to the best of our knowledge.


\section{Model and Preliminaries}
We study the pricing problem faced by a central seller controlling a set $S$ of goods with a large number of buyers, each belonging to one of the buyer types in $B$. All the buyers having a given type $B_i$ have the same set of desired items $S_i \subseteq S$. We model the market structure as a bipartite graph $G=(B\cup S, E)$ where there is an edge between each buyer type $B_i$ and every good in $S_i$. For every individual buyer $j \in B_i$, her valuation is $v_j$ for items in her demand set $S_i$ and $0$ otherwise. Note that different buyers belonging to the same type $B_i$ can have different valuations for the items in $S_i$.

\noindent\textbf{Aggregate Demand and Production Cost:}
Every individual buyer's demand is infinitesimal compared to the market size. Therefore, we can model the aggregate demand of all buyers having type $B_i$ using an inverse demand function $\lambda_i(x)$; $v=\lambda_i(x_i)$ means that $x_i$ of these buyers have a value of $v$ or more for the items in $S_i$. As an example, consider $\lambda_i(x) = 1-x$ for $x \in [0,1]$. This means that the total population of buyers with type $B_i$ is one; $\lambda_i(0.25) = 0.75$ implies that one-fourth of these buyers have a valuation of $0.75$ or more. Finally, the seller incurs a production cost of $C_t(x)$ for producing $x$ amount of good $t \in S$.

\noindent\textbf{Best-Reponse and Envy-Freeness:} A complete solution consists of prices and an allocation, and is specified by three vectors $(\vec{p}, \vec{x}, \vec{y})$. The seller's strategy is to select a price vector $\vec{p}$ where $p_t$ is the price on item $t \in S$. We define $\vec{x}$ to be the buyer demand vector such that $x_i$ is the amount of good allocated to buyers from type $B_i$. Finally, $\vec{y}$ is the allocation such that $y_t$ is the total amount of good $t$ allocated to buyers and $y_t(i)$, the amount to buyer type $i$. We only consider allocations $\vec{y}$ that are \emph{feasible} with $\vec{x}$ and $G$: for all $i$, $\sum_{t} y_t(i)$ should equal $x_i$, and buyers in $B_i$ must only receive allocations of items belonging to $S_i$. Then,

\begin{itemize}
\item Given $\vec{p}$, we let $\bi$ denote the minimum price available to buyers from type $B_i$, i.e., $\bi = \min_{t \in S_i} p_t.$

\item The buyer demand $\vec{x}$ is said to be a \textbf{best-response} to the prices $\vec{p}$ iff $\forall B_i$, $\bi=\lambda_i(x_i)$. That is, a population of $x_i$ buyers from $B_i$ have a value of $\bi$ or larger, and thus are maximizing their utility by deciding to purchase items at a price of $\bi$.

\item Given $\vec{p}$ and $\vec{x}$, the allocation $\vec{y}$ is said to be \textbf{envy-free} if buyer demand is a best-response to the prices, and if for every buyer the items they are allocated are the lowest priced items available to them, i.e., $y_t(i) > 0 \Rightarrow p_t = \bi$.
\end{itemize}

Our main objective is an envy-free solution that maximizes revenue. Given $(\vec{p}, \vec{x}, \vec{y})$, the {\em revenue} of the seller is the total payment minus costs incurred, i.e.,
\begin{align*}
\text{Revenue} = & \sum_{t \in S}(p_ty_t - C_t(y_t))\\
 = & \sum_{i \in B}\bar{p}_i x_i - C(\vec{y}) & \text{(if solution is envy-free)}
\end{align*}
Note that as long as the instance is clear, we will use $C(\vec{y})$ to denote the total cost of all items when the allocation is $\vec{y}$. We also consider solutions with good social welfare, i.e., the total utility of all the buyers plus that of the seller. As long as the solution is envy-free, buyers are utility-maximizing, and so the aggregate utility of buyers belonging to type $i$ is the sum of their values minus payments, which is $\int_{0}^{x_i}\lambda_i(x)dx - \bi x_i.$ Since the payments cancel out, the total social welfare of a solution is equal to
$$\text{Social Welfare} = \sum_{B_i \in B}\int_{x=0}^{x_i}\lambda_i(x)dx - C(\vec{y}).$$

 We make the following assumptions on the inverse demand and cost functions.
\begin{enumerate}
\item $\lambda_i(x)$ cannot increase with $x$. This is by definition: if $x_1$ buyers hold a value of $v_1$ or more and $x_2 < x_1$, then the value of these $x_2$ buyers is at least $v_1$. 

\item We also assume that $\lambda_i(x)$ is \emph{continuously differentiable} on $(0,T_i)$ (here $T_i$ is the population of buyers in $B_i$), and has a monotone hazard rate (see definition below). Notice that $\lambda'_i(x)$ cannot be positive since $\lambda_i$ is non-increasing.

\item For all $t \in S$, we take the production costs $C_t(y)$ to be convex, which is the norm in the literature. In addition, we assume that $C_t(y)$ is continuously differentiable and define $c_t(y)$ to be its derivative. All our results also hold if an item $t$ has a limited supply of $Y_t$, and $C_t(y)$ is only differentiable until $y=Y_t$, at which point it becomes infinite.

\end{enumerate}

\begin{definition}{\emph{(MHR)}} An inverse demand function $\lx$ is said to be log-concave or equivalently, have a monotone hazard rate if $\frac{\ldx}{\lx}$ is non-increasing with $x$. Since $\ldx$ is not positive, this is equivalent to saying $\frac{|\ldx|}{\lx}$ is non-decreasing.
\end{definition}
Many commonly used buyer demand functions belong to this class including uniform ($\lx = a$), linear ($\lx = a - x$) and exponential inverse demand ($\lx = e^{-x}$). Although the monotone hazard rate requirement gives the appearance of being somewhat restrictive, this assumption is actually rather weak. We show (proof in the Appendix) that even with only MHR demand, our framework encompasses the well-studied unit-demand pricing problem in finite markets.

\begin{proposition}
\label{prop_mhrnottoobad}
Any UDP instance with uniform valuations in markets with a finite number of buyers can be reduced to an instance of our problem where all buyer types have monotone hazard rate inverse demand.
\end{proposition}

Therefore, our setting strictly generalizes previously studied UDP problems, which are unlikely to admit sub-logarithmic approximation factors~\cite{briest2011buying, chalermsook2012improved}. Our main contribution, however, is proving that the addition of a little bit of structure (via uniform peak or support) to our general framework provides much greater insight into the nature of the revenue-maximizing solution, and leads to good algorithms.

{\bf Optimal Solutions.~}
We use the notation $(\vec{p^{opt}},\vec{x^{opt}}, \vec{y^{opt}})$ to denote an envy-free solution maximizing revenue, and $(\vec{x^*}, \vec{y^*})$ to denote an allocation that maximizes welfare (since welfare depends only on the allocation, not the prices). Given a graph $G$, functions $\lambda_i$ and $C_t$, it is easy to see that the solution maximizing social welfare can be computed using a convex program. We also remark here that once the welfare maximizing solution is computed, there exist prices $\vec{p^*}$ so that $(\vec{x^*}, \vec{y^*})$ is an envy-free allocation to these prices. The more challenging task is to compute prices that (approximately) maximize revenue and perhaps, simultaneously welfare. 

\begin{proposition}
\label{lem_optbrmincost}
Consider the optimum solution $\vec{x^*}, \vec{y^*}$ for a given instance. Define the price vector $\vec{p^*}$ as $p^*_t = c_t(y^*_t)$ for item $t$. Then $(\vec{p^*},\vec{x^*}, \vec{y^*})$ is an envy-free solution to the prices. 
\end{proposition}
We prove this in the Appendix and also show that in the revenue maximizing solution, every item's price is at least its price in $\vec{p^*}$, i.e, for all $t$, $p^{opt}_t \geq p^*_t$. For the rest of this paper, we will only consider solutions where the prices dominate $\vec{p^*}$. In fact, in Lemma~\ref{lem_lowerboundmain} we show much stronger lower bounds on $\vec{p^{opt}}$ which enable us to prove our results.
%

\textbf{Connection to Flows:} We can view a feasible allocation $\vec{y}$ as a flow from the items $S$ to the buyers with a demand of $\vec{x}$, assuming that $G$ is fixed. Notice that there are several feasible flows for a given demand $\vec{x}$. We will be most interested in min-cost flows: the feasible allocation $\vec{y}$ that also minimizes the total production cost $\sum_t C_t(y_t)$. The min-cost flow is independent of the prices and, given $\vec{x}$, can be computed efficiently using a convex program.

It is easy to see that $\vec{y^*}$ is a min-cost flow, but general envy-free solutions, including  $(\vec{p^{opt}},\vec{x^{opt}}, \vec{y^{opt}})$, may not use min-cost flows, since envy-freeness constrains the buyers to use only the items with cheapest {\em price,} while min-cost flows form allocations to optimize production costs.
We reiterate that given a price vector $\vec{p}$, the best-response buyer demand $\vec{x}$ can be computed using $\bi = \lambda_i(x_i)$, and given $(\vec{p}, \vec{x})$, we can always determine an envy-free allocation $\vec{y}$. Interestingly, the solutions returned by our algorithms are not only envy-free, but also will use min-cost flows for the corresponding buyer demand $\vec{x}$. Finally, the proof of our 1.88-Approximation Algorithm crucially uses the following property that relates best-response allocations and min-cost flows.
\begin{lemma}
\label{lem_diffcostmonoton}
Consider two price vectors $\vec{p_1}$ and $\vec{p_2}$ such that $\vec{p_1} \geq \vec{p_2}$. Let $\vec{x_1}, \vec{x_2}$ be the corresponding best-response buyer demands to these prices and $\vec{y_1}$, $\vec{y_2}$ be the minimum-cost flows for the buyer demands $\vec{x_1}$ and $\vec{x_2}$ respectively. Then,
\begin{enumerate}
\item $\vec{x_2} \geq \vec{x_1}$, i.e, every buyer type's demand is higher under $\vec{p_2}$ than under $\vec{p_1}$.

\item For all $t \in S$, $c_t(y^1_t) \leq c_t(y^2_t)$.
\end{enumerate}
\end{lemma}
\emph{(Proof Sketch)} The lemma merely formalizes a very intuitive idea, namely that increasing prices from $\vec{p_2}$ to $\vec{p_1}$ can only lead to lowered buyer demand. Since the buyer demand in $\vec{x_1}$ is smaller, it means that the allocation of any item cannot strictly increase for a min-cost flow. Finally, a smaller allocation implies a smaller marginal cost. Rigorously proving Statement 2 is actually not that easy and we defer the full proof to the appendix. $\blacksquare$

\section{Large Markets with Uniform Peak Valuations}\label{sec:uniformpeak}
As argued in the Introduction, for markets with a large number of buyers it often makes sense to assume that the inverse demand functions $\lambda_i$ have the same support $[\lambda^{min}, \lambda^{max}]$ for all $i$, which is what we do in this section. In fact, all our results hold as long as the {\em peak} values for every $\lambda_i$ are the same, i.e., that $\lambda^{\max}=\lambda_i(0)$ is the same for all $i$. This would occur, for example, when a very large population of buyers is assigned to different buyer types in a random way. Not too surprisingly, the problem of revenue maximization remains NP-Hard even when the demand functions have uniform peak valuations.

\begin{proposition}
\label{prop_nphard}
The Unit Demand Pricing problem in large markets with MHR Inverse Demand and Uniform Peaks $(\lambda_i(0) = \lambda^{max}$ for all $i$) is NP-Hard even with unlimited supply.
\end{proposition}

In this section we establish our main result: a 1.877 approximation algorithm for maximizing revenue, which works as long as the inverse demand functions are MHR and have uniform peak values. We begin with a general, parameter-dependent procedure for generating prices, which will be a building block of both this algorithm, and the algorithms in later sections. Although the algorithm is described here as a more intuitive continuous-time procedure, it can be efficiently implemented using $O(|B|\log \lambda^{\max})$ min-cost flow computations, as we argue in Section~\ref{sec:efficient}. To simplify discussion, henceforth we will use ``buyer" interchangeably with ``buyer type" as long as the context is clear.

\begin{algorithm}[htbp]
\caption{Ascending-Price Procedure with Stop Parameter $k$}
\label{alg_genprocedurebody}
\algsetup{indent=2em}
\begin{algorithmic}[1]
\STATE Set initial prices on the items, $p_t = p^*_t$.
\STATE ACTIVE $\gets$ All minimally priced items and all buyers using these items.
\STATE INACTIVE $\gets B\cup S$ $\setminus$ ACTIVE; FINISH $\gets \emptyset$
\WHILE{FINISH $\neq B \cup S$}
\STATE Increase the price of all ACTIVE items by an infinitesimal amount\\
\COMMENT{All ACTIVE items have the same price, the active price.}
\STATE Compute the min-cost flow for the sub-graph induced by ACTIVE
\COMMENT{We prove later: At every stage active buyers only receive allocations of active items}
\IF{$t \in$ INACTIVE s.t $p^*_t$ equals the active price}
\STATE Remove $t$, buyers using $t$ from INACTIVE and add to ACTIVE
\ENDIF
\IF {$t \in$ ACTIVE meets the stopping criterion in the current solution}
\STATE Remove $t$, buyers using $t$ from ACTIVE and add to FINISH.
\ENDIF
\ENDWHILE
\begin{equation}
\label{eqn_stopconditionmain}
\textbf{Stopping Criterion($p_t, y_t, k$)}:~~~~p_t - c_t(y_t) \geq \frac{1}{k}(\lambda_{max} - c_t(y_t))
\end{equation}
\end{algorithmic}
\end{algorithm}

Algorithm \ref{alg_genprocedurebody} begins by pricing all the items at the price vector $\vec{p^*}$, which as we mentioned makes $(\vec{p^*}, \vec{x^*}, \vec{y^*})$ an envy-free allocation. We gradually increase prices on the items belonging to an `active set', initialized to the set of cheapest items in $\vec{p^*}$ and the buyers receiving these items. At each stage, every item in the active set has the same price (active price) allowing us to compute the min-cost flow for only the active buyers and items. As we increase the active price, if it equals $p^*_t$ for some inactive $t$, we add $t$ and buyers using $t$ to the active set. An item $t$ remains active until a stopping condition dependent on a parameter $k \geq 1$ is reached (Equation~\ref{eqn_stopconditionmain}); once this happens the price of item $t$ becomes fixed, and item $t$ is removed from the active set along with buyers using $t$. The simple ascending-price algorithm stops once all the items have met the stopping criterion.

We now sketch some properties of this algorithm that hold for all $k$; the full proofs and formal lemmas are in the Appendix. For a given parameter $k \geq 1$, we will use $(\vec{p^{k}}, \vec{x^{k}}, \vec{y^{k}})$ to denote the solution returned by our algorithm. We begin with an important observation regarding the stopping criterion: at any stage of the algorithm, for any two active items, if the item with the higher marginal cost ($c_t(y_t)$) meets the stopping criterion, then the item with the smaller marginal cost must also satisfy the condition (Lemma~\ref{lem_ep1stopcrithier} in the Appendix).

\textbf{Notation}
Every `stage' of our algorithm corresponds to a unique value of the active price (i.e., price of all the active items), so we can refer to the allocation formed by the algorithm at some point as the allocation at active price $p$. Formally, we define $\vec{x}(p)$ to be the buyer demand vector when the active price is $p$, and $\vec{y}(p)$ is the allocation of items at that price. At any price $p$, for the inactive buyers $x_i(p)$ coincides with $x^*_i$ and for inactive items $y_t(p) = y^*_t$. For finished items (or buyers), the allocation (demand) is the same as it was when that item (buyer) met the stopping criterion and became finished. Finally, we use $B_A(p)$ to denote the buyers and $S_A(p)$ for the items in the active set when the active price is $p$; we define the analogous sets $(B_I(p), S_I(p))$, $(B_F(p), S_F(p))$ for the inactive and finished blocks respectively.

\subsection*{Properties Satisfied by Algorithm~\ref{alg_genprocedurebody} and price hierarchy}

\begin{SCfigure}
\label{fig_sethierarchy}
\includegraphics[scale=0.8]{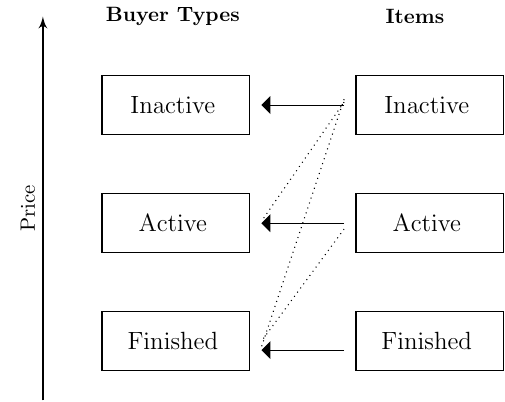}
\caption{At any stage of Algorithm~\ref{alg_genprocedurebody}, the order of prices for the different items is: Inactive $>$ Active $>$ Finished. Thick edges indicate that buyers in a certain set (Active, Inactive or Finished) receive allocations only from the items in the same set. Dotted edges between a buyer set and an item set indicate although the buyers have access to the items in that set, they are not currently receiving any allocation of that item.
}
\end{SCfigure}

Figure~\ref{fig_sethierarchy} describes the natural hierarchy between the Active, Inactive, and Finished sets at every value of the active price $p$. It is not difficult to show that the statement in Figure~\ref{fig_sethierarchy} always holds, starting with the initial envy-free solution $(\vec{p^*}, \vec{x^*}, \vec{y^*})$. We prove this formally in the Appendix. Our next claim shows that as we increase the active price, the marginal cost of items in the active set cannot increase. 

\begin{lemma}
\label{lem_marginaldec}
Suppose that some item $t$ belongs to the active set at two different active prices $p_1$ and $p_2$ with $p_1 < p_2$, then
$c_{t}(y_t(p_1)) \geq c_{t}(y_t(p_2))$.
\end{lemma}
\emph{(Proof Sketch)} The result relies heavily on Lemma~\ref{lem_diffcostmonoton} and Corollary~\ref{corr_diffcostmorebuyers}. When we increase the active price, the buyer demand decreases and thereby the allocation of items and the marginal cost decreases. 
$\blacksquare$ 

The next proposition gives us additional insight regarding the stopping condition. It tells us that every item $t$ actually meets the stopping criterion at equality and therefore, the greater than or equal to sign in Equation~\ref{eqn_stopconditionmain} is redundant.

\begin{proposition}
\label{prop_stopequality}
For any given item $t$ and fixed $k$, the stopping condition must be obeyed at equality. Formally, suppose that $t$ meets the stopping criterion at an active price of $p$, then
$$p - c_t(y_t(p)) = \frac{1}{k}(\lambda^{max} - c_t(y_t(p)).$$
\end{proposition}
\emph{(Proof Sketch)} 
For any item $t$, its initial price $p^*_t = c_t(y^*_t)$ and so initially, the LHS of Equation~\ref{eqn_stopconditionmain} is zero and the RHS is non-zero, so the LHS is strictly smaller. The initial price for any item cannot be larger than $\lambda^{max}$. So, if the active price is $\lambda^{max}$, the LHS must be greater than or equal to the RHS. It stands to reason therefore, that every item meets the stopping criterion at some intermediate price. Moreover, if we increase the price by a small amount, then the marginal cost can also decrease only by a small amount due to the reduced demand. Since, both the LHS and the RHS in the above equation change continuously, they must be equal at some point. $\blacksquare$

We also remark here that at every stage of the algorithm, for any buyer $i$, all the items she uses at that stage must have the same marginal cost. Given an allocation $\vec{y}$, we will use $r_i(\vec{y})$ to denote the (unique) marginal cost of the items being used by buyer $i$ in that allocation as long as all the items she uses have the same marginal. We are now ready to prove our first main result regarding our algorithm. We show that for any value of $k$, the stopping parameter, the solution returned by our algorithm is an envy-free allocation. In addition, the allocation $\vec{y^k}$ is also the minimum cost flow for the demand $\vec{x^k}$.

\begin{theorem}
\label{thm_ep1envyfree}
For any given value of $k$, Algorithm \ref{alg_genprocedurebody} returns prices $\vec{p^k}$ and an envy-free allocation $\vec{x^k}, \vec{y^k}$. Moreover $\vec{y^k}$ is also the minimum cost flow corresponding to the demand $\vec{x^k}$.
\end{theorem}
\begin{proof}
Recall that an allocation is envy-free if all buyers only purchase from the minimally priced items available to them, i.e., $\forall (i,t) \in E$, $\bar{p}^k_i \leq p^k_t$. An allocation $\vec{y^k}$ is a min-cost flow for the demand $\vec{x^k}$ if and only if $\forall (i,t) \in E$, $r_i(\vec{y}) \leq c_t(y^k_t)$. That is, in addition to using the minimally priced items buyers are also using the items with the smallest marginal costs.

Assume by contradiction that the allocation is not envy-free. Then, for some $(i,t) \in E$, $\bar{p}^k_i > p^k_t$. This means that when $t$ reached the stopping criterion at price $p^k_t$, $i$ was either active or inactive. Now consider some price $p \in (p^k_t, \bar{p}^k_i)$. At this price $t \in S_F(p)$ and $i \in B_A(p) \cup B_I(p)$. However, by Figure~\ref{fig_sethierarchy}, there can be no edge between $i$ and $t$ which is a contradiction. Therefore, the allocation returned by our algorithm is indeed envy-free.

Next, suppose that the allocation is not a minimum cost flow. Then $\exists (i,t) \in E$ such that $c_t(y^k_t) < r_i(\vec{y^k})$. We claim that this implies $p^k_t < \bar{p}^k_i$. Rearranging the equation in the statement of Proposition~\ref{prop_stopequality}, we get that (remember $\bar{p}_i^k = p^k_{t'}$ for some item $t'$ that $i$ receives and $r_i(\vec{y^k}) = c_{t'}(y^k_{t'})$)
$$\bar{p}^k_i = \frac{1}{k}(\lambda^{max} + (k-1)r_i(\vec{y^k})) > \frac{1}{k}(\lambda^{max} + (k-1)c_t(y^k_t)) = p^k_t.$$

This means that $\exists (i,t)\in E$ with $p^k_t < \bar{p}^k_i$, which violates the envy-free condition. Therefore, the allocation is also the minimum cost flow satisfying the given demand.
$\blacksquare$ \end{proof}

The most crucial lemma that allows us to prove our approximation factor is the following lower bound which we prove on the prices in the revenue maximizing solution $\vec{p^{opt}}$. Unlike most existing work, this lower bound allows us to compare our solution directly to the revenue-maximizing solution, instead of using the welfare-maximizing solution as a proxy.

\begin{lemma}
\label{lem_lowerboundmain}
For MHR inverse demand functions, the price of every item $t$ in the profit-maximizing solution $\vec{p^{opt}}$ is at least its price in $\vec{p^e}$, i.e., $p_t^{opt}\geq p^e_t$.
\end{lemma}
\begin{proof} 
The proof proceeds as follows. We first show that in any solution where some items are priced below their price in $\vec{p^e}$, a few of these items do not meet the stopping criterion at $k=e$. Then, we show that for any monotone hazard rate demand function that does not satisfy the stopping criterion at $k=e$, we can always increase the price on the items and improve the profits thereby contradicting the optimality of $\vec{p^{opt}}$.

Assume by contradiction that in the optimal solution some items have a price smaller than their price in $\vec{p^e}$. Let $S_{min}$ be the subset of such items with the smallest price (call it $p_{min}$). Since the optimum solution is envy-free and our solution is a min-cost flow, we can apply Lemma~\ref{lem_mixedpricediffcost}. As per the lemma, there must exist some $t \in S_{min}$ such that $c_{t}(y^e_t) \leq c_{t}(y_t^{opt})$. Call this item $t_{min}$.

Construct a directed graph $G$' whose vertices are the same as in $G$ but with the following edges $E'$
\begin{enumerate}
\item $(t,i) \in E'$ if $i$ is receiving non-zero amounts of item $t$ in $\vec{y^{opt}}$.
\item $(i,t) \in E'$ if $(i,t) \in E$ and $p^{opt}_t = \bar{p}^{opt}_i$, i.e., $t$'s price coincides with the price of the cheapest item available to $i$.
\end{enumerate}

Let $S^+_{min}$ be the set of items that reachable from $t_{min}$ and $B^+_{min}$ be the set of buyers reachable from $t_{min}$ in this graph $G$'. We make three simple observations here: first, for every item $t \in S^+_{min}$, its price must equal $p_{min}$. Second, every buyer in $B^+_{min}$ is only receiving allocations of the items in $S^+_{min}$ and has no edge in $E$ to any item outside of $S^+_{min}$ also priced at $p_{min}$. Finally, for every item $t \in S^+_{min}$, its marginal cost in OPT is at least $c_{t_{min}}(y^{opt}_{t_{min}})$.

These three observations imply that in a revenue maximizing solution, for the reduced instance with only the buyers and sellers in $B^+_{min}$ and $S^+_{min}$, for the corresponding demand in $OPT$, the sub-allocation on these items must be a min-cost flow. Since $p^e_{t_{min}} > p_{min}$, $t_{min}$ cannot satisfy the stopping criterion ($k=e$) based on its price and allocation at $OPT$, i.e.,

$$p_{min} - c_{t_{min}}(y^{opt}_{t_{min}}) < \frac{1}{e} (\lambda^{max} - c_{t_{min}}(y^{opt}_{t_{min}})).$$

 Moreover, for every other $t \in S^+_{min}$, its price in $OPT$ is $p_{min}$ and marginal cost is at least as much as that of $t_{min}$. Therefore,

\begin{equation}
\label{eqn_notsatisfystop}
p_{min} - c_{t}(y^{opt}_{t}) < \frac{1}{e} (\lambda^{max} - c_{t}(y^{opt}_t))
\end{equation}

Now, our idea is the following: we will uniformly increase the price on only the items in $S^+_{min}$ by a sufficiently small amount so that the buyers from $B^+_{min}$ still use only these items in an envy-free solution. Then we will use the stopping criterion to show that at the new price, the seller's profit strictly increases thereby violating the fact that $OPT$ is a revenue-maximizing solution.

We let $\hat{p}$ denote the smallest price in $\vec{p^e} \cup \vec{p^{opt}}$ that is strictly larger than $p_{min}$. It is clear that as long as we increase the price of all $t \in S^+_{min}$ to some $p \in [p_{min}, \hat{p})$, the cheapest items for buyers in $B^+_{min}$ will only come from $S^+_{min}$. Moreover, for any other buyer $i \notin B^+_{min}$, the set of cheapest items will not change.

Now, gradually increase the price of only the items from $S^+_{min}$, compute the min-cost allocation for the buyers in $B^+_{min}$ using only these items. Retain the same price and allocation for every other buyer and item. At any $p \in  [p_{min}, \hat{p})$, denote by $\tilde{c}(p)$, the smallest marginal cost of any item in $S^+_{min}$ at the new allocation at price $p$. Define a price $p^+$ based on one of two cases,

\begin{enumerate}
\item At some minimal $p'$ in the domain $(p_{min}, \hat{p})$, the following condition is met,
$$p' - \tilde{c}(p') = \frac{1}{e}(\lambda^{max}- \tilde{c}(p')).$$

Recall that the above condition is not met at $p=p_{min}$. Then, set $p^+= p'$.

\item At no $p \in [p_{min}, \hat{p})$ is the above condition met. Set $p^+ = \frac{1}{2}(\hat{p} + p_{min})$.
\end{enumerate}

We remark that if some item meets the stopping condition above at price $p^+$ at all, then it must be the item(s) whose marginal cost equals $\tilde{c}(p^+)$ (See Propostion~\ref{lem_ep1stopcrithier}).

Define $\vec{p^+}$ as the price vector where items in $S^+_{min}$ are priced at $p^+$ and the rest retain their price in $OPT$. Let $\vec{x^+}$ and $\vec{y^+}$ be the corresponding buyer demand and envy-free solution. Our main claim is the following: the profit $\pi^+$ at $(\vec{p^+}, \vec{x^+}, \vec{y^+})$ is larger than the optimal profit $\pi^*$, which is a contradiction. Consider the difference between the two profits (note that the payments and cost remains the same for buyers and items not in $B^+_{min}$ and $S^+_{min}$ respectively).

\begin{align*}
\pi^+ - \pi^* = & \sum_{i \in B^+_{min}}(\bar{p}^+_i x^+_i - \bar{p}^{opt}_i x^{opt}_i) - (C(\vec{y^+}) - C(\vec{y^{opt}}))\\
\geq & \sum_{i \in B^+_{min}}(\bar{p}^+_i x^+_i - \bar{p}^{opt}_i x^{opt}_i) - \sum_{t \in S^+_{min}}(c_t(y^+_t)(y^+_t - y^{opt}_t)) \\
\geq & \sum_{i \in B^+_{min}}(\bar{p}^+_i x^+_i - \bar{p}^{opt}_i x^{opt}_i) - \sum_{t \in S^+_{min}}(\tilde{c}(p^+) (y^+_t - y^{opt}_t))\\
= & \sum_{i \in B^+_{min}}(\lambda_i(x^+_i) - \tilde{c}(p^+))x^+_i - \sum_{i \in B^+_{min}}(\lambda_i(x^{opt}_i) - \tilde{c}(p^+))x^{opt}_i
\end{align*}

The first inequality comes from observing that $\vec{p^{+}}$ dominates $\vec{p^{opt}}$ and then applying Corollary~\ref{corr_flowmagnitudemincost}. The final equality is from rearranging the allocation from the items to the buyers and using the fact $\bar{p}^+_i$ and $\bar{p}^{opt}_i$ are simply equal to the respective $\lambda_i$ values.

Now we make a strong claim: that for all $i \in B^+_{min}$, $(\lambda_i(x^+_i) - \tilde{c}(p^+))x^+_i - (\lambda_i(x^{opt}_i) - \tilde{c}(p^+))x^{opt}_i > 0$. Clearly this would imply that $\pi^+ > \pi^*$, thereby completing the contradiction. So for the rest of the proof, we will focus on showing this claim.

Essentially the claim follows from the following two nice properties that hold for any non-increasing MHR function $f_i(x)$.
\begin{enumerate}
\item (Lemma~\ref{lem_subclaim_mhr}) If $f_i(0) \geq ef_i(x)$ for some $x$, then $\frac{|f'_i(x)|}{f_i(x)} \geq \frac{1}{x}$.
\item (Lemma~\ref{mhr_profitchanges}) If $\frac{|f'_i(x_1)|}{f_i(x_1)} \geq \frac{1}{x_1}$ and $x_2 > x_1$, then $f(x_1)x_1 > f(x_2)x_2$.
\end{enumerate}

We show how the above two properties give us the desired claim. Define for all $i$, $f_i(x) = \lambda_i(x) - \tilde{c}(p^+)$. Clearly, this function still has a monotone hazard rate since $\lambda_i$ is MHR. Now, by definition of $p^+$, we know that $f_i(0) \geq f_i(x^+_i)$, i.e,

$$\frac{1}{e}(\lambda^{max} - \tilde{c}(p^+)) \geq (p^+ - \tilde{c}(p^+)).$$

Therefore, from Lemma~\ref{lem_subclaim_mhr}, we can conclude that
$$\frac{|f'_i(x^+_i)|}{f_i(x^+_i)} \geq \frac{1}{x^+_i}.$$

Now, we use this in the second lemma with $x_1 = x^+_i$ and $x_2 = x^{opt}_i$. We know $x^{opt}_i > x^+_i$. Therefore, we get, $f_i(x^+_i)x^+_i > f_i(x^{opt}_i)x^{opt}_i.$ Replacing $f_i$ with the actual definition, we get the desired result

$$(\lambda_i(x^+_i) - \tilde{c}(p^+))x^+_i > (\lambda_i(x^{opt}_i) - \tilde{c}(p^+))x^{opt}_i.$$ \end{proof}

We now describe our actual algorithm to approximately maximize profit that uses the general procedure described in Algorithm~\ref{alg_genprocedurebody}. The algorithm is reasonably straightforward. We make two calls to Procedure~\ref{alg_genprocedurebody} for $k=e$ and $k=\sqrt{e}$. 

\begin{algorithm}[htbp]
\caption{$0.53$-Approximate Algorithm to Maximize Profit}
\label{alg_uniformpeak}
\algsetup{indent=2em}
\begin{algorithmic}[1]
\STATE Let $\pi_1$ be the profit of the solution returned by Algorithm~\ref{alg_genprocedurebody} for $k=e$.
\STATE Let $\pi_2$ be the profit of the solution returned by Algorithm~\ref{alg_genprocedurebody} for $k=\sqrt{e}$.
\STATE Return $\max(\pi_1, \pi_2)$ and its corresponding prices and allocation.
\end{algorithmic}
\end{algorithm}

\begin{theorem}
\label{thm:1.88approx}
Algorithm~\ref{alg_uniformpeak} returns an envy-free allocation which is a $(4\sqrt{e}-2-e)\approx 1.877$ approximation to the optimal profit.
\end{theorem}
\emph{(Proof Sketch)}
Since we already argued that Algorithm \ref{alg_genprocedurebody} returns envy-free solutions, we only need to establish the approximation bound. We also claim that $\vec{p^{\sqrt{e}}} \geq \vec{p^e}$; this is proved in the Appendix. 

Define $B^H$ to be the buyers whose payment in $\vec{p^{opt}}$ is larger than in $\vec{p^{\sqrt{e}}}$, and $B^L$ the buyers whose payments are between $\vec{p^e}$ and  $\vec{p^{\sqrt{e}}}$. We show in the Appendix that $\vec{p^e}$ extracts a large fraction of optimum profit from the buyers in $B^L$ and $\vec{p^{\sqrt{e}}}$ from $B^H$. A key lemma that completes the bound is that for MHR functions, for an increase in price from $\vec{p^e}$ to $\vec{p^{\sqrt{e}}}$, the profit loss is at most a factor two. Therefore, $\vec{p^{\sqrt{e}}}$ extracts at least half the profit from the buyers in $B^L$. The precise factor of 1.88 comes from carefully balancing these bounds; this leads to the choice of $\vec{p^{\sqrt{e}}}$ and $\vec{p^e}$. $\blacksquare$

It is important to note that $\vec{p^{\sqrt{e}}}$ is not simply a scaled version of the prices in $\vec{p^e}$; its construction crucially depends on the stopping condition, which in turn depends on both the price and the production cost. The presence of production costs means that previous approaches (e.g., scale prices uniformly, choose a single price for all items) do not work well, as they can end up with solutions with high production cost and thus low overall profit.

\subsection{Approximating Revenue and Social Welfare Simultaneously}
For sellers who care about both revenue and welfare, as is common in repeated mechanisms where you want the buyers to ``leave happy", we also provide the following guarantees.

\begin{theorem}\label{thm:welfare}
Algorithm~\ref{alg_genprocedurebody} for $k=e$ provides an envy-free solution which is a $e$-approximation to the optimal profit with at least half the optimal welfare.
\end{theorem}
\begin{proof}
The first part, bounding the profit, is rather easy. We simply refer to the Proof of Theorem~\ref{thm:1.88approx} where we used $\pi_1$ to denote the profit from the $k=e$ solution. We have already shown that

$$\pi_1 \geq \frac{1}{e}(\pi^{opt}(B^H)) + \frac{1}{\sqrt{e}}(\pi^{opt}(B^L)).$$

This means that $\pi_1 \geq \frac{1}{e}(\pi^{opt}(B^H) + \pi^{opt}(B^L)) \geq \frac{1}{e}(\pi^*)$, and so the profit returned by the algorithm is at most a factor $e$ smaller than the optimal profit. We now move on to the social welfare. The social welfare of our solution and the optimum are as given below,

$$SW(\vec{x^e}, \vec{y^e}) = \sum_{i \in B} \int_{x=0}^{x_i^e} \lambda_i(x)dx - \sum_{t \in S}C_t(y^e_t).$$

\begin{align*}
SW(\vec{x^*}, \vec{y^*}) = & \sum_{i \in B} \int_{x=0}^{x_i^*} \lambda_i(x)dx - \sum_{t \in S}C_t(y^*_t)\\
= & SW(\vec{x^e}, \vec{y^e})+ \sum_{i \in B} \int_{x=x_i^e}^{x_i^*} \lambda_i(x)dx - \sum_{t \in S}(C_t(y^*_t) - C_t(y^e_t))\\
= & SW(\vec{x^e}, \vec{y^e}) + \text{Welfare Loss}.
\end{align*}

For the rest of the proof, we will attempt to bound the lost welfare in terms of the social welfare of our solution. In particular, we will show that the lost welfare for MHR functions cannot be any larger than the welfare of our solution, which will give us the half approximation. We know that for every $i$ the following is true for $k=e$ due to Proposition~\ref{prop_stopequality} (recall that $\lambda_i(x^e_i)=\bar{p}_i^e$):

$$\lambda_i(x^e_i) - r_i(\vec{y^e}) = \frac{1}{e}(\lambda^{max} - r_i(\vec{y^e})).$$

Look at the function $\lambda_i(x) -  r_i(\vec{y^e})$: since the latter term is a constant, we know that this function has a monotone hazard rate. Applying the contrapositive of Lemma~\ref{lem_subclaim_mhr}, we get that for all $i$,
$$\frac{\lambda_i(x^e_i) - r_i(\vec{y^e})}{|\lambda'_i(x^e_i)|} \leq x^e_i.$$

Next, we claim that the total difference in production costs at the optimum and our solution is at least $\sum_i r_i(\vec{y^e})(x^*_i - x^e_i)$. This is formally shown in Lemma~\ref{lem_costdiff} in the Appendix. Therefore, the following is an upper bound for the Lost Welfare:
\begin{align*}
\text{Welfare Loss} \leq & \sum_{i \in B} \int_{x=x_i^e}^{x_i^*} \lambda_i(x)dx - \sum_i r_i(\vec{y^e})(x^*_i - x^e_i)\\
= & \sum_{i \in B} \int_{x=x_i^e}^{x_i^*}  [\lambda_i(x) - r_i(\vec{y^e})]dx.
\end{align*}

For every $i$, the second term inside the integral is a constant and so the function inside the integral also has a monotone hazard rate in the desired interval. This means that $\forall x \in [x^e_i, x^*_i]$,

$$\frac{\lambda_i(x) - r_i(\vec{y^e})}{|\lambda'_i(x)|} \leq \frac{\lambda_i(x^e_i) - r_i(\vec{y^e})}{|\lambda'_i(x^e_i)|} \leq x^e_i.$$

So we can bound every integral as follows,
\begin{align*}
\int_{x^e_i}^{x^*_i}[\lambda_i(x) -r_i(\vec{y^e})] dx \leq & \int_{x^e_i}^{x^*_i}x^e_i |\lambda'_i(x)|dx\\
= & x^e_i \int_{x^e_i}^{x^*_i}(-\lambda'_i(x))dx \\
= & x^e_i(\lambda_i(x^e_i) - \lambda_i(x^*_i)).
\end{align*}

\noindent Now consider $\lambda_i(x^*_i)$. Since our solution $\vec{p^e}$ dominates $\vec{p^*}$, it is not hard to see that $\lambda_i(x^*_i) \geq r_i(\vec{y^*}) \geq r_i(\vec{y^e})$ (Lemma~\ref{lem_diffcostmonoton}). So, we finally bound the lost welfare as follows:

\begin{align*}
\text{Lost Welfare} \leq & \sum_{i \in B} \int_{x=x_i^e}^{x_i^*}  [\lambda_i(x) - r_i(\vec{y^e})]dx\\
\leq & \sum_{i \in B} x^e_i(\lambda_i(x^e_i) - \lambda_i(x^*_i))\\
\leq & \sum_{i \in B} \lambda_i(x^e_i) x^e_i - x^e_i r_i(\vec{y^e}) \\
\leq & \sum_{i \in B}\bar{p}^e_i x^e_i - \sum_t C_t(y^e_t) \\
= & \pi_1\\
\leq & SW(\vec{x^e}, \vec{y^e}).
\end{align*}
The last step is true because for any given solution, the profit cannot be larger than the social welfare of the same solution. So the optimum social welfare is $ SW(\vec{x^e}, \vec{y^e}) + \text{Lost Welfare}$ which is no larger than $2\cdot SW(\vec{x^e}, \vec{y^e})$. This completes the proof. \end{proof}

This result provides an additional, stronger revenue-welfare trade-off. Suppose we run the algorithm in Theorem \ref{thm:welfare}, and obtain welfare which is exactly $\frac{1}{\alpha}$ of optimum (we know that $\alpha\leq 2$). Then, our analysis guarantees that the profit of the resulting solution is actually at least $\max(\frac{1}{e},\frac{\alpha-1}{\alpha})$ of optimum; for instance if $\alpha = 2$, then we actually get half the optimal revenue.

\begin{corollary}
$\exists \alpha \in [1,2]$ such that the solution returned by Theorem~\ref{thm:welfare} has a fraction $\frac{1}{\alpha}$ of the optimum welfare and $\max(\frac{1}{e}, \frac{\alpha-1}{\alpha}))$ of the optimum revenue.
\end{corollary}

\section{Efficient Implementation of the Ascending-Price Procedure}
\label{sec:efficient}
We now describe how to implement Algorithm~\ref{alg_genprocedurebody} efficiently. The following algorithm uses \\$O(|B|\log(\lambda^{max}))$ min-cost computations, where $|B|$ is the number of buyer types and $\lambda^{max}$ is the peak of the buyer valuation functions. The efficient implementation depends crucially on the following fact: Fix any active set $(B_1, S_1)$ and an active price range $[p_1, p_2]$. Look at only min-cost allocations. Suppose that at $p_1$ (all active items have this price), none of the items in the active set meet the stopping criterion (for some $k$) and at some $p \in (p_1, p_2]$, an item $t$ meets the stopping criterion. Then, if we compute the minimum cost allocation/flow when all the active items are priced at $p_2$, $t$ must still meet the stopping criterion. This property hints at a `binary search'-like approach to identify the exact price at which an item meets the stopping criterion. We first prove the property and then provide the algorithm.

\begin{lemma}
\label{lem_prop_binarysearch}
Consider some set of (active) buyers and items $(B_1, S_1)$ and suppose that when all items are priced at $p$, some item $t$ meets the stopping criterion in Equation~\ref{eqn_stopconditionmain} for the corresponding min-cost allocation $(\vec{x^a}, \vec{y^a})$. For any $p_2 \geq p$, let $\vec{x^b}, \vec{y^b}$ be the min-cost best-response solution when all of the items in $S_1$ are priced at $p_2$. Then $t$ must meet the stopping criterion at this new allocation, i.e.,
$$p_2 - c_t(y^b_t) \geq \frac{1}{k}(\lambda^{max} - c_t(y^b_t)).$$
\end{lemma}
\begin{proof}
The proof is straightforward and follows from an application of Lemma~\ref{lem_diffcostmonoton}. Since $\vec{x^a}$ dominates $\vec{x^b}$, by the lemma, we know that $c_t(y^b_t) \leq c_t(y^a_t)$. We also know from the stopping criterion at $p$ that
$$p \geq \frac{1}{k}(\lambda^{max} + (k-1)c_t(y^a_t)) \geq \frac{1}{k}(\lambda^{max} + (k-1)c_t(y^b_t)).$$
$p$ is smaller than $p_2$ and therefore, $t$ also meets the stopping criterion for the flow $\vec{y^b}$.
$\blacksquare$ \end{proof}

We need some additional notation before we define the algorithm. Consider the set of prices $\vec{p^*}$ given by marginal cost pricing at the optimum solution. This is the starting point for Algorithm~\ref{alg_genprocedurebody}. Define the boundary price vector $P = \left\lbrace p_0, p_1, \cdots, p_m \right\rbrace$, such that $p_0$ is the unique smallest price in $\vec{p^*}$, $p_1$ is the unique second smallest price in $\vec{p^*}$ and so on. Finally set $p_m = \lambda^{max}$. For all $j < m$, when the active price is $P(j)=p_{j-1}$, some new items and buyers enter the active set because their initial price was also $P(j)$. We first describe the algorithm semi-formally and then show correctness.
\begin{enumerate}
\item Initialize the active set to be the same as the initial active set in Algorithm~\ref{alg_genprocedurebody}.
\item Iterate for $j=0$ to $m$.
\item Set the price of active items to be $P(j)$ and compute a min-cost best-response allocation for the active buyers, $(\vec{x}(j), \vec{y}(j))$.
\item Let $S_F(j)$ be the set of items that meet the stopping condition at this allocation.
\item Run a binary search for the items in $S_F(j)$ in the interval $[P(j-1), P(j)]$ and remove them at the exact price at which they become finished along with the buyers using them.
\item Add items to the active set whose price in $\vec{p^*}$ is $P(j)$ and the buyers using them in $\vec{y^*}$.
\item Let $B_A(j)$ be the new set of active buyers and $S_A(j)$, the active items.
\item Repeat the process.
\end{enumerate}

Before showing correctness of the above algorithm and elucidating upon the `binary search' in Step 6, we show some simple invariants of the above algorithm that prove that the above algorithm `simulates' Algorithm~\ref{alg_genprocedurebody}. The proof is in the Appendix.

\begin{lemma}
\label{lem_complementbinarysearch}
The following invariants hold during the course of the above Algorithm.
\begin{enumerate}
\item For any $j$, $B_A(j) \cup S_A(j) = B_A(P(j)) \cup S_A(P(j))$.

\item For any $j$, all the items in $S_F(j)$ meet the stopping criterion in the interval $[P(j-1), P(j)]$ during the course of Algorithm~\ref{alg_genprocedurebody}.
\end{enumerate}
\end{lemma}
Recall that $B_A(P(j)) \cup S_A(P(j))$ denote the contents of the active set in Algorithm~\ref{alg_genprocedurebody} when the active price was $P(j)$. We already know any item that meets the stopping criterion in the interval should show up in $S_F(j)$. The above lemma complements this result by saying all items that show up in $S_F(j)$ must meet the stopping criterion. Conditional upon Invariant 1 holding up to some iteration $j-1$ and invariant 2 holding up to iteration $j$, we now explain the binary search procedure before proving the invariants.

\textbf{Binary Search:} Consider $S_F(j)$ and let $t \in S_F(j)$ be the item that reaches the stopping condition first in $[P(j-1), P(j)]$ in the original algorithm. Clearly for all $p' > p$ in that interval, it must meet the stopping criterion and for $p' < p$, no item in $S_F(j)$ could have met the stopping criterion. Therefore, we can effectively use binary search to identify $p$. Now we can repeat this for all $t' \in S_F(j)$ in the reduced interval $[p, P(j)]$.

\section{Relaxing the Uniform Peak Valuation Assumption}
In this section, we relax the assumption that for all demand functions, $\lambda_i(0)$ is the same. We capture the distortion in this quantity via a parameter $\Delta$ which is the ratio of the maximum value of $\lambda_i(0)$ over all $i$ to the minimum. Even though the $\lambda_i$'s may not be the same, it is likely that they are closely distributed if all buyer types are interested in a similar type of good. Our next result shows that in such markets, we can still extract a good fraction of the optimum revenue and welfare. For this result, we also require that each cost function $C_t(x)$ is doubly convex, i.e., its derivative $c_t(x)$ is also convex with $c_t(0) =0$.

\begin{theorem} 
\label{thm_generalmhr}
For any instance with MHR Demand and Doubly Convex Costs, we can compute an envy-free solution which has a $O(\log \Delta)$-approximation to the optimal revenue and which also guarantees $\frac{1}{4}^{th}$ of the optimum welfare.
\end{theorem}

Moreover, we also prove in the Appendix that this result is actually tight under mild complexity assumptions and that the doubly convex assumption is required. 

\begin{proposition}
\label{clm_complexitygeneralmhr}
\begin{enumerate}
\item There cannot be a $O(\Delta^{k})$-approximation algorithm for any $k > 0$ for UDP in Large Markets with MHR Inverse Demand and Convex Costs (instead of doubly convex) unless $NP \subseteq DTIME(n^{(log^{c}n)})$ for some constant $c$.

\item There is no constant factor approximation algorithm for our UDP problem in large markets with MHR inverse demand and doubly convex costs unless $NP \subseteq DTIME(n^{(log^{c}n)})$ for some constant $c$.
\end{enumerate}
\end{proposition}

\section{Conclusion}
In this paper, we considered envy-free pricing in very large markets. Our results suggest that, unlike in markets with few buyers, very good pricing schemes can be computed efficiently for such large markets. For example, if the seller wants to maximize his revenue, our algorithm provides prices which result in a 1.88 approximation to maximum revenue. If the seller cares about both revenue and welfare (as is often the case in repeated interactions), then our second result provides a pricing scheme which results in provably high revenue {\em and} welfare.

\bibliography{bibliography}
\bibliographystyle{plain}

\appendix

\section{More Preliminary Results and Proofs from Section 2}

\subsection*{Additional Notation}

\begin{prop_app}{prop_mhrnottoobad}
Any UDP instance with uniform valuations in markets with a finite number of buyers can be reduced to an instance of our problem where all buyer types have monotone hazard rate inverse demand.
\end{prop_app}
\begin{proof}
First consider a UDP-Uniform-Valuations (UDP-UV) problem with unlimited supply on the items. Let $B'$ be the set of buyers with unit demand and valuation $v_i$ and $S'$ be the set of items. Reducing this to an instance of our problem, create one buyer type for every $i \in B'$ such that $\lambda_i(x) = v_i$ for $x \leq 1$ and $0$ otherwise. Clearly $\lambda_i$ is continuously differentiable in $(0,1)$ and is MHR in that interval, so these functions do follow our framework. Next the items are the same as $S'$ and the cost functions are zero everywhere.

Consider an solution of $UDP-UV$, clearly the solution is feasible for our instance and has the same value of the objective function. Consider a solution of our problem, we show that $\exists$ a solution of the UDP-UV with equal or larger revenue. Suppose that some buyer is receiving a positive but fractional allocation, then we can simply increase the allocation of this buyer to one leading to a solution with increased profit. If $\exists$ some $i \in B$ sending flow to more than one item, we can simply transfer all of the flow to any one of these items without any decrease in profit. Therefore, we can convert any allocation to an \emph{integral} allocation and therefore, a feasible solution for UDP-UV with the same prices. Therefore, the optima must also coincide.

Next, what if we have limited supply? We have already shown how to model these limited supply functions using cost. That is, we can set $C_t(x) = 0$ for $x \leq \gamma_t$ and $C_t(x) = \infty$ for $x > \gamma_t$, where $\gamma_t$ is the capacity of the item. Alternatively, we can also take $C_t(x)$ to be zero until $x$ gets very close to $\gamma_t$ at which point $C_t$ starts increasing very fast upto a large number at $C_t(\gamma_t)$. We assume that the capacities or supply of every item is integral in the original UDP-UV instance or else we round it down while forming our instance. So any solution for UDP-UV is still a feasible solution for our problem. Consider any solution for our problem with finite cost (we can safely ignore the solutions where capacities are violated). Once again, suppose that some buyer is receiving a fractional allocation. Either we can increase the allocation to this buyer to one without violating the supply constraints or since all the supplies are integral, we can increase the allocation of this buyer to one, reducing the allocations of some other buyers with fraction allocation without any change in profit. In this way, we can convert the solution to one every buyer is receiving an integral amount of the good. Similarly, we can rearrange flow on the items so that every buyer is receiving one unit of one single good without violating capacity constraints (integral capacities imply integral flows). \end{proof}


\begin{prop_app}{lem_optbrmincost}
Consider the optimum solution $\vec{x^*}, \vec{y^*}$ for a given instance. Define the price vector $\vec{p^*}$ as $p^*_t = c_t(y^*_t)$ for item $t$. Then $(\vec{p^*},\vec{x^*}, \vec{y^*})$ is an envy-free solution to the prices. Moreover, $\vec{y^*}$ is the min-cost flow for buyer demand $\vec{x^*}$.
\end{prop_app}
\begin{proof}
Look at the optimum solution and the prices $p^*_t=c_t(y^*_t)$. Every buyer has to necessarily send flow on the cheapest items available to her. Indeed, assume by contradiction that some buyer $i$ has non-zero flow on item $t'$ and access to item $t$ such that $p^*_t < p^*_{t'}$. Then $c_t(y^*_t) < c_{t'}(y^*_{t'})$. Therefore, one can shift some infinitesimal flow corresponding to buyer $i$ from $t'$ to $t$ and reduce the cost of the optimal solution, a contradiction. Also, notice if $\exists$ another flow satisfying the same buyer demand but with a smaller cost, then we can use that flow and reduce the cost of the optimal solution. So $\vec{y^*}$ is a min-cost flow. \end{proof}

\begin{lemma}
\label{lem_optisgood}
Let $(\vec{p^{opt}}, \vec{x^{opt}}, \vec{y^{opt}})$ be the revenue-maximizing solution. Then for every $t$, $p^{opt}_t \geq p^*_t$. That is, the prices at the welfare maximizing solution provide a weak lower bound for optimal prices.
\end{lemma}
\begin{proof}
Assume by contradiction that in the optimal solution some items have a price strictly smaller than their price in $\vec{p^*}$. Let $S_{min}$ be the subset of such items with the smallest price (call it $p_{min}$). As per Lemma~\ref{lem_mixedpricediffcost}, $\exists$ $t_{min} \in S_{min}$, such that $c_t(y^*_t) \leq c_t(y^{opt}_t)$.

Construct a directed graph $G'$ whose vertices are the same as in $G$ but with the following edges $E'$
\begin{enumerate}
\item $(t,i) \in E'$ if $i$ is receiving non-zero amounts of item $t$ in $\vec{y^{opt}}$.
\item $(i,t) \in E'$ if $(i,t) \in E$ and $p^{opt}_t = \bar{p}^{opt}_i$, i.e., $t$'s price coincides with the price of the cheapest item available to $i$
\end{enumerate}

Let $S^+_{min}$ be the set of items that are reachable from $t_{min}$ and $B^+_{min}$ be the set of buyers reachable from $t_{min}$ in this graph $G'$. We make two simple observations here: first, for every item $t \in S^+_{min}$, its price must equal $p_{min}$. Second, every buyer in $B^+_{min}$ is only receiving allocations of the items in $S^+_{min}$ and has no edge in $E$ to any item outside of $S^+_{min}$ priced at $p_{min}$.

We know that in the welfare maximizer, for every item $t$, $p^*_t = c_t(y^*_t)$. But we know that for $t_{min}$, $p_{min} < p^*_{t_{min}}$. Therefore for every item $t \in S^+_{min}$, $p^{opt}_t = p_{min} <  p^*_{t_{min}} = c_{t_{min}}(y^*_{t_{min}}) \leq c_t(y^{opt}_t)$. This cannot be a good sign for any profit maximizing solution because the price has to be at least the marginal cost, otherwise the seller can increase his price, lower the marginal cost and improve profits. We show this more formally.

Let $x(p)$ be the total demand from the buyers in $B^+_{min}$ at a price of $p$. Look at the revenue-maximizing solution and increase the price of \textbf{only} the items in $S^+_{min}$ by a sufficiently small $\epsilon$. Clearly, for a small enough $\epsilon$, the best response of all buyers in $B^+_{min}$ is to still send a flow on $S^+_{min}$. Moreover, for a small enough increase in price, buyer demand can only decrease by a small amount.

Recall that in the allocation $\vec{y^{opt}}$, we showed that for all $t\in S_{min}^+$, we have $c_t(y_t^{opt})>p_{min}$. Since the marginal costs on all items are continuous with the allocation, it is not hard to reason that for a sufficiently small increase in price, for every $t$, the marginal cost on every item $t \in S^+_{min}$ after the price increase still cannot be smaller than $p_{min} + \epsilon$. Let the new allocation on every item be $y^2_t$. Then, the difference in profit from OPT is due to only these items and buyers and is,
\begin{align*}
= & (p_{min} + \epsilon)(x(p_{min}+\epsilon)) - p_{min}x(p_{min}) - \sum_{t \in S^+_{min}}(C_t(y^2_t) - C_t(y^{opt}_t)\\
\geq & (p_{min} + \epsilon)(x(p_{min}+\epsilon)) - p_{min}x(p_{min}) - \sum_{t \in S^+_{min}}\left(c_t(y^2_t)(y^{opt}_t - y^2_t))\right)\\
\geq & (p_{min} + \epsilon)(x(p_{min}+\epsilon)) - p_{min}x(p_{min}) - (p_{min}+\epsilon)(x(p_{min}+\epsilon) - x(p_{min}))\\
\geq & \epsilon x(p_{min}) > 0.
\end{align*}
This is a contradiction with the fact that $OPT$ maximizes the profit.  \end{proof}

\subsection*{Important Properties of Min-Cost Flows}
We now show some nice properties of min-cost flows that we will use extensively in the following sections. For a flow $\vec{y}$, we define $r_i(\vec{y})$ to be the minimum marginal cost $c_t(y_t)$ of all items received by buyer $i$ (i.e., with $y_t(i)>0$). Due to KKT conditions, if $\vec{y}$ is a min-cost flow, then $i$ is only allocated items with marginal cost equal to the minimum marginal cost of any item available to $i$, i.e., for a min-cost flow $r_i(\vec{y})=\min_{(i,t)\in E}c_t(y_t)$. Given an allocation $\vec{y}$, we will also use $C(\vec{y})$ to denote the total (production) cost of all the items as long as the instance is clear.

Finally, for the rest of the Appendix, we are only concerned about solutions where the prices are at least $\vec{p^*}$. This is for obvious reasons since we know by Lemma~\ref{lem_optisgood} that the revenue maximizing prices cannot be smaller than $\vec{p^*}$. We call this the weak price lower bound assumption.

\begin{assumption}{(Weak Price Lower Bound Assumption for given $\vec{p})$}
\label{assumption_weakpricelowerbound}
$$ \vec{p} \geq \vec{p^*}.$$
\end{assumption}

\begin{lem_app}{lem_diffcostmonoton}
Consider buyer demand vectors $\vec{x^1}$ and $\vec{x^2}$ such that $\vec{x^2}$ dominates $\vec{x^1}$, i.e., $\vec{x^2} \geq \vec{x^1}$ componentwise. Let $\vec{z^1}$ and $\vec{z^2}$ be the min-cost flows corresponding to the two demands respectively. Then for all items $t$, $c_t(z^1_t) \leq c_t(z^2_t)$. Moreover, for every buyer $i$, $r_i(\vec{z^1}) \leq r_i(\vec{z^2})$.
\end{lem_app}
\begin{proof}
We only need to prove the first part of the lemma, i.e., for all $t$, $c_t(z^1_{t}) \leq c_t(z^2_{t})$. Once we show this, the second part follows almost directly. Indeed, suppose that a buyer $i$ received some quantity of item $t$ in $\vec{z^2}$, then whatever item she receives in $\vec{z^1}$ has to have a marginal that is smaller or equal to that of $t$. We know that the marginal cost of $t$ in $\vec{z^1}$ is not larger than that in $\vec{z^2}$.

We now proceed to prove our main claim by contradiction, suppose for some item $t$, $c_t(z^1_t) > c_t(z^2_t)$. Since the marginal cost function is monotone non-decreasing, this must mean that $z^1_t > z^2_t$.

Now let us construct the following graph $G'=(S, E')$ where $S$ is the set of all items. We say that there is a directed edge from item $t_1$ to $t_2$ if $\exists$ some buyer $i$ such that
$$z^1_{t_1}(i) > z^2_{t_1}(i) \text{ and } z^1_{t_2}(i) < z^2_{t_2}(i).$$

In simple terms, this means that $i$ is receiving more amount of $t_2$ and less of $t_1$ in $\vec{z^2}$ than what she received in $\vec{z^1}$. This also means that $i$ is receiving non-zero amounts of $t_1$ in $\vec{z^1}$ and $t_2$ in $\vec{z^2}$. Now, look at item $t$. Since the total allocation of this item is smaller in $\vec{z^2}$, this must mean that there is at least one buyer $i$ who is sending less flow on $t$ in $\vec{z^2}$ as compared to before. However, the total demand of $i$ is only larger in $\vec{x^2}$, which means there must be some other item $t_1$ to which she is sending more flow than before. This implies that $(t,t_1) \in E'$.

Suppose that $S_t$ represents the set of items that are reachable from $t$ in $G'$ including $t$ itself. We have already shown that $S_t$ has at least one item other than $t$. Our first claim is that all the nodes in $S_t$ have a marginal cost in $\vec{z^2}$ that is no larger than the marginal cost of $t$ in the same allocation. To show this consider an edge $(t_1 , t_2)$ where both the items belong to $S_t$. By definition, there must be some buyer who has access to both these items and is sending non-zero flow on $t_2$ in $\vec{z^2}$. Since is a min-cost allocation, it means the marginal cost of $t_2$ in $\vec{z^2}$ cannot be larger than that of $t_1$. Applying this transitively from $t$, all nodes reachable from $t$ must have a marginal cost smaller than or equal to $c_t(z^2_t)$.

Similarly, for $(t_1, t_2) \in E'$, both belonging to $S_t$, some buyer has non-zero flow on $t_1$ in $\vec{z^1}$ and this must imply that $c_t(z^1_t) \leq c_{t_1}(z^1_{t_1})$ for all $t_1 \in S_t$. Using these inequalities regarding the marginal costs in $\vec{z^1}$ and $\vec{z^2}$, we get for all $t_1 \in S_t$,
\begin{equation}
\label{eqn_diffcostinequalities}
c_{t_1}(z^2_{t_1}) \leq c_{t}(z^2_{t}) < c_t(z^1_t) \leq c_{t_1}(z^1_{t_1}).
\end{equation}

What this means is that for all the items in $S_t$, the incoming flow is larger in $\vec{z^1}$ as compared to $\vec{z^2}$. Suppose that $B^1_t$ is the complete set of buyers who receive non-zero amounts of the items in $S_t$ in $\vec{z^1}$. Our final claim is that every buyer in $B^1_t$ receives more or equal amount of the items in $S_t$ in $\vec{z^2}$ as compared to $\vec{z^1}$. That is for buyers in $B^1_t$,
$$\sum_{t_1 \in S_t}z^1_{t_1}(i) \leq  \sum_{t_1 \in S_t}z^2_{t_1}(i).$$

Notice that for any buyer $i$ if this is not true, then there must exist at least one $t_1 \in S_t$ which she receives more in $\vec{z^1}$ than $\vec{z^2}$. However buyer $i$'s total demand has increased in $\vec{x^2}$ but the consumption from $S_t$ has decreased and so there must be some $t_3$ outside of $S_t$ to which she sends more flow in $\vec{z^2}$ than $\vec{z^1}$. But this means that there must be an edge from $t_1$ to $t_3$ and so $t_3 \in S_t$, a contradiction.

Now, we are ready to prove our main result. Recall that $\forall t_1 \in S_t$, $z^1_{t_1} > z^2_{t_1}$. Since for all $t_1 \in S_t$, the incoming flow in $\vec{z^1}$ can only come from the buyers in $B^1_t$.
\begin{align*}
\sum_{t_1 \in S_t}z^1_{t_1} = & \sum_{i \in B^1_t} \sum_{t_1 \in S_t}z^1_{t_1}(i) \\
\leq & \sum_{i \in B^1_t} \sum_{t_1 \in S_t}z^2_{t_1}(i)\\
\leq & \sum_{t_1 \in S_t}z^2_{t_1} \\
< & \sum_{t_1 \in S_t}z^1_{t_1}.
\end{align*}
This is a contradiction. \end{proof}

\begin{corollary}
\label{corr_flowmagnitudemincost}
Consider two price vectors $\vec{p}$ and $\vec{p'}$ such that $\vec{p'} \geq \vec{p}$ for every component. Let $(\vec{p},\vec{x},\vec{y})$ be a solution in which $x$ is a best-response demand vector to prices $\vec{p}$, and $\vec{y}$ is a min-cost flow with demand $\vec{x}$. Let $(\vec{p'},\vec{x'},\vec{y'})$ be the similar solution for prices $\vec{p'}$. Then, for every $t$, either $y'_t < y_t$ or $c_{t}(y_t) = c_t(y'_t)$.
\end{corollary}

\begin{proof}
Higher prices imply that the demand vector is smaller. This allows us to apply Lemma \ref{lem_diffcostmonoton}, and the rest follows due to these being min-cost flows. \end{proof}

\begin{corollary}
\label{corr_diffcostmorebuyers}
Consider an instance $(B_1 \cup S, E_1)$ where the buyers have a demand vector $\vec{x^1}$ and the corresponding min-cost flow is $\vec{z^1}$. Let $(B_2 \cup S, E_1 \cup E_2)$ be another instance with $B_1 \subseteq B_2$ and $E_2$ only has edges between $B_2 \setminus B_1$ and $S$. Let $\vec{x^2}$ be some demand vector for this instance and $\vec{z^2}$ is the min-cost flow for this demand such that $\vec{x^2}$ dominates $\vec{x^1}$, i.e., $\forall i \in B_1, x^2_i \geq x^1_i$. Then for all $t \in S$, $c_t(x^1_t) \leq c_t(x^2_t)$.
\end{corollary}
\begin{proof}
We can simply reduce the first instance to another instance where the set of buyers is $B_2$, edges $E_1 \cup E_2$ but the demand for the additional buyers is zero. Now, the corollary reduces to Lemma~\ref{lem_diffcostmonoton}. \end{proof}

\begin{lemma}
\label{lem_flowpartition}
\textbf{(Flow Partition Lemma)} Consider an instance $(B,S)$ with buyer demand $\vec{x}$ and corresponding min-cost flow $\vec{y}$. Suppose that $B^H \subseteq B$ and let $\vec{x}(B^H)$ denote the demand sub-vector for this subset and $\vec{y}(B^H)$ be the min-cost flow for this demand sub-vector alone (i.e., when the buyer set is $B^H$). Then, we can partition the total cost as
$$C(\vec{y}) = C(\vec{y}(B^H)) + \left(C(\vec{y}) - C(\vec{y}(B^H))\right).$$
Moreover, the following must be true
\begin{enumerate}
\item $\forall t$, either $y_t(B^H) \leq y_t$ or $c_t(y_t(B^H)) = c_t(y_t)$.
\item $\sum_{i \in B \setminus B^H}r_i(\vec{y})x_i \geq \sum_{t \in S}c_t(y_t)(y_t - y_t(B^H)) \geq \sum_{t \in S}\left(C_t(y_t) - C_t(y_t(B^H))\right)$.
\end{enumerate}
\end{lemma}
\begin{proof}
The partitioning of the cost is a trivial result (just add and subtract $C_t(y_t(B^H))$ for all $t$). Moreover, Point 1 in the second part of the lemma follows almost directly from Corollaries~\ref{corr_diffcostmorebuyers} and ~\ref{corr_flowmagnitudemincost}. So, we only focus on proving point 2.

Now for any given $t$, suppose that $y_t \geq y_t(B^H)$. Then, by the convexity of $C_t$ it follows that $C_t(y_t) - C_t(y_t(B^H)) \leq c_t(y_t)(y_t - y_t(B^H))$. If $y_t < y_t(B^H)$, then we know that the marginal cost is the same at both these flows. Therefore, $C_t(y_t) - C_t(y_t(B^H)) = c_t(y_t)(y_t - y_t(B^H))$. Summing this up over all $t$, we get one half of Result (Point) 2.

Now,
\begin{align*}
\sum_{t}c_t(y_t)(y_t - y_t(B^H)) =& \sum_{t}c_t(y_t)y_t - \sum_t c_t(y_t)y_t(B^H)\\
= & \sum_{i \in B}r_i(\vec{y})x_i - \sum_t c_t(y_t) \sum_{i \in B} y_t(B^H)(i).
\end{align*}

Note that $y_t(B^H)(i) > 0$ only for $i \in B^H$ and only if $i$ is receiving non-zero amount of item $t$ in $\vec{y}(B^H)$. This means that $i$ has an edge to $t$ and therefore in $\vec{y}$, $r_i(\vec{y}) \leq c_t(y_t)$. Therefore, some flow rearrangement gives us
$$\sum_t \sum_{i \in B}c_t(y_t) y_t(B^H)(i) \geq \sum_t \sum_{i \in B}r_i(\vec{y}) y_t(B^H)(i) = \sum_{i \in B^H}r_i(\vec{y})x_i(B^H).$$
And so we finally get
$$\sum_{t}c_t(y_t)(y_t - y_t(B^H)) \leq \sum_{i \in B}r_i(\vec{y})x_i - \sum_{i \in B^H}r_i(\vec{y})x_i,$$
which gives us Point 2. \end{proof}

\begin{proposition}
\label{prop_ignoreprofit}
(Ignoring Profit Propostion)
Consider any price vector $\vec{p} \geq \vec{p^*}$ componentwise. Let $\vec{x}$ be the corresponding B-R buyer demand and let $\vec{y}$ be the min-cost flow for $\vec{x}$. Then, for any desired subset $B_1$ of the buyers $B$,

$$\pi_1 = \sum_{i \in B}\bi x_i - C(\vec{y}) \geq \sum_{i \in B_1}\bi x_i - C(\vec{y}(B_1)),$$

where $\vec{y}(B_1)$ is the min-cost flow for the reduced demand by only the buyers in $B_1$.
\end{proposition}
That is, we set the demand to be zero for all buyers outside $B_1$ and compute the min-cost flow for the same instance.
\begin{proof}
Since the price dominates $\vec{p^*}$, $\vec{x^*}$ has to dominate $\vec{x}$, i.e., increasing price on all items can only lead to a drop in buyer demand. Applying Lemma~\ref{lem_diffcostmonoton}, we get that for all $i$, $c_t(y_t) \leq c_t(y^*_t) = p^*_t$. Now applying the flow partition lemma (Lemma \ref{lem_flowpartition}), we get,

$$\pi_1 = \sum_{i \in B_1}\bi x_i - C(\vec{y}(B_1)) + \sum_{i \notin B_1}\bi x_i - \sum_{i \notin B_1}r_i(\vec{y}) x_i.$$

But we know that $r_i(\vec{y}) \leq r_i(\vec{y^*}) = \bar{p}^*_i \leq \bi$. Therefore, the second term in the above inequality is non-negative and the first term is a lower bound for $\pi_1$. \end{proof}

\section{Proofs from Section 3}
In this section, we will consider a market where every buyer's\footnote{We use buyer and buyer type interchangeably from now on.} inverse demand function $\lambda_i(x)$ has a monotone hazard rate and the same value of $\lambda_i(0)$. Formally,
\begin{assumption}
\label{assumption_uniformpeak}
There exists some $\lambda^{max} > 0$ such that for every buyer $i$, $\lambda_i(0) = \lambda^{max}$.
\end{assumption}

\begin{prop_app}{prop_nphard}
The Unit Demand Pricing problem in large markets with MHR Inverse Demand and Uniform Peaks $(\lambda_i(0) = \lambda^{max}$ for all $i$) is NP-Hard even with zero production costs.
\end{prop_app}
\begin{proof}
We just sketch the proof here since the general idea is the same as the hardness proof in~\cite{guruswami2005profit}. Consider an instance of vertex cover. Reducing this to our problem, there is one buyer for each vertex and $r$ buyers for every edge $e$ such that $r = |V|=n$.  Moreover, there is also one item in $S$ for every vertex of the original instance All vertex buyers have unit demand with valuation $2$, i.e., $\lambda_i(x) = 2$ for $x \leq 1$ and $0$ otherwise. All the edge buyers have the following $MHR$ inverse demand function $\lambda_i(x) = 2-x$. The vertex buyers have access only to the corresponding vertex item and the edge buyers have access to items corresponding to its two end points. 

First, it is not hard to see that the maximum profit from the set of $r$ buyers for an edge is $r$ (each of these can give only a max profit of $1$ when $p=1$). Next, in any solution of our problem, the items whose prices are less than $2$ must form a vertex cover of the original graph. Moreover, the price of these sellers is at most $1+\frac{1}{2r}$. Say you are given some solution where there are $k$ nodes priced below $2$, then the profit of this solution is at least $2(n-k) + k + mr$, where $m$ is the number of edges. This is obtained by pricing all these nodes at $1$. Next, the profit of this solution is at most, $2(n-k) + k(1+\frac{1}{2r}) + mr$. Now, it is not hard to show that any optimal solution must be a minimal vertex cover.\end{proof}

Recall our notation from Section~\ref{sec:uniformpeak}, $\vec{p^k}, \vec{x^k}, \vec{y^k}$ is the solution returned by Algorithm~\ref{alg_genprocedurebody} for a given value of $k$.

\subsection*{Properties Regarding the Stopping Criterion}

The following simple lemma shows that for a fixed value of $k$, if there are two items and the one with the larger marginal cost meets the stopping criterion, then the other must also satisfy the condition.

\begin{proposition}
\label{lem_ep1stopcrithier}
Consider two items $t$, $t'$ at the same price $p$ and let $\vec{y}$ be some min-cost allocation such that $c_t(y_t) \leq c_{t'}(y_{t'})$. If $t'$ satisfies the stopping criterion, then $t$ must also satisfy the stopping criterion.
\end{proposition}
\begin{proof}
The proof follows from a rearrangement of Equation~\ref{eqn_stopconditionmain}. The stopping criterion for $t'$ can also be written as
$$\lambda^{max} \leq kp - (k-1)c_{t'}(y_{t'}).$$
The term in the RHS is in turn no larger than $kp - (k-1)c_{t}(y_{t})$, which implies that $t$ also satisfies the stopping criterion. \end{proof}

Our second proposition compares the prices of items in the solution returned by our algorithm for two different values of $k$. We say that vector $\vec{a} \geq \vec{b}$ iff each element of $\vec{a}$ is not smaller than its corresponding element of $\vec{b}$.

\begin{proposition}
\label{prop_stoptwovalues}
Suppose that $k_1, k_2 \geq 1$ are two values of the stopping parameter with $k_1\geq k_2$. Then, $\vec{p^{k_2}} \geq \vec{p^{k_1}}$ and $\vec{x^{k_1}} \geq \vec{x^{k_2}}$.
\end{proposition}
\begin{proof}
Notice that the initial prices $\vec{p^*}$ are independent of $k$. This means that whatever be the value of the stopping parameter $k$, an item enters the active set at the exact same value of the active price. Also, note that when an item becomes finished, all the buyers having an edge to this item must also be finished (or else why is it not using this item?). Suppose that during Algorithm~\ref{alg_genprocedurebody} for $k=k_2$, the items become finished in the order $O=(t_1, t_2, t_3, \ldots, t_{|S|})$, breaking ties arbitrarily. Then, clearly $p^{k_2}_{t_1} \leq p^{k_2}_{t_2} \leq \ldots \leq p^{k_2}_{t_{|S|}}$. Assume by contradiction that $t$ is the item with the smallest index in $O$ such that $p^{k_2}_t < p^{k_1}_t$.

Consider the run of Algorithm~\ref{alg_genprocedurebody} for $k=k_1$ and let $A(p^{k_2}_t)$ be the buyers and items in the active set when the active price is $p^{k_2}_t$. Clearly $t$ belongs to the set but does not meet the stopping criterion yet. We claim that for any buyer $i \in A(p^{k_2}_t)$, for the state of Algorithm~\ref{alg_genprocedurebody}($k=k_2$) at active price $p^{k_2}_t$, $i$ belonged to the active set. Moreover, all the items in $A(p^{k_2}_t)$ must have been active for the same active price but $k=k_2$. Therefore, applying Corollary~\ref{corr_diffcostmorebuyers} comparing the contents of the active set at $k=k_1$ and $k=k_2$, we conclude that the marginal cost of $t$ at $k=k_1$ cannot be larger than its marginal cost for $k=k_2$ at the same active price. Therefore, it is not hard to see that for $k=k_1$, $t$ must satisfy the stopping condition at active price $p^{k_2}_t$. \end{proof}

\noindent\textbf{Price-Hierarchy:} The next proposition is a trivial observation from the definition of the algorithm, which we state without proof. It states that at a given active price $p$, the price of items in $F(p)$ are no larger than $p$ and the items in $I(p)$ cannot be smaller than $p$.

\begin{proposition}
\label{lem_ep1pricehierarchy}
Suppose at some given active price $p$, $t_1, t_2, t_3$ are three items such that $t_1 \in A(p)$, $t_2 \in I(p)$ and $t_3 \in F(p)$. Then,
$$P_{t_3}(p) \leq P_{t_1}(p) = p \leq P_{t_2}(p).$$
\end{proposition}
We now establish some easy invariants that hold during the course of our algorithm. In particular, we show that there can be no edges between buyers in $B_I(p)$ and items in $S_A(p)$ or $S_F(p)$ and between buyers in $B_A(p)$ and items in $S_F(p)$. 

\begin{proposition}
\label{prop_envyfreeinvariants}
At any active price $p$, the following must be true
\begin{enumerate}
\item Suppose user $i$ belongs to $B_I(p)$ and item $t \in S_A(p)$ or $t \in S_F(p)$. Then, there cannot be an edge between user $i$ and item $t$.

\item Suppose user $i \in B_A(p)$ and item $t \in S_F(p)$. Then, there cannot be an edge between user $i$ and item $t$.
\end{enumerate}
\end{proposition}
The invariants indicate a natural hierarchy in the partitions as shown in Figure~\ref{fig_sethierarchy}. Users who are inactive can only have edges to items that are inactive. Users that are active can only have edges to items that are active or inactive.
\begin{proof}
We prove these invariants by contradiction. First, suppose that for some $i \in B_I(p)$ and $t \in S_A(p)$, $(i,t) \in E$. Since $i$ is inactive, this must mean that every item $t'$ that $i$ used in the optimum solution must have an initial price $c_{t'}(y^*_{t'}) > p$, where $p$ is the active price. But since $t \in S_A(p)$, this must mean that at some $p' \leq p$, the active price must have been equal to the initial price $c_t(y^*_t)$ of the item $t$. So, $c_t(y^*_t)$ is smaller than $c_{t'}(y^*_{t'})$ where $t'$ is some item used by $i$ in the optimum solution. Thus, if $i$ had access to $t$, then we could have shifted an infinitesimal amount of flow to $t$ and reduced the cost of the optimum solution, which is a contradiction. The proof for the case when $t \in S_F(p)$ is similar since if a item is finished, then it must have been active at some lower price.

For the second invariant, assume that there is some $i \in B_A(p)$ and item $t \in S_F(p)$ such that $i$ has access to item $t$. Let $p'$ be the price where $t$ met the stopping criterion and was transferred to the finished set. If at this price, $i$ used $t$, then by definition, $i$ would have also been added to FINISH. This means that either $i \in B_I(p')$ or $i$ did not send any flow on item $t$. $i \in B_I(p')$ contradicts the first invariant. Now suppose, $t'$ is some item that was allocated to user $i$ at active price $p'$. Then since we computed a min-cost allocation inside the active set, $c_{t'}(y_{t'}(p')) \leq c_t(y_t(p'))$. However, by Proposition~\ref{lem_ep1stopcrithier}, $t'$ also meets the stopping criterion and thus $i$ would have also been transferred to FINISH, a contradiction.\end{proof}

\begin{lem_app}{lem_marginaldec}
Suppose that some item $t$ belongs to the active set at two different active prices $p_1$ and $p_2$ with $p_1 < p_2$, then
$c_{t}(y_t(p_1)) \geq c_{t}(y_t(p_2))$.
\end{lem_app}
\begin{proof}
Let's begin by considering the set of items and buyers in the active set at price $p_2$, i.e., buyers $B_A(p_2)$ and items $S_A(p_2)$. Suppose that the demand of only these buyers at price $p_2$ is $\vec{x^2}$ and their corresponding allocations from the items in $S_A(p_2)$ are $\vec{y^2}$. By definition, $\vec{y^2}$ is a min-cost flow for $\vec{x^2}$ for the sub-instance $(B_A(p_2), S_A(p_2))$. Moreover for every item $t \in S_A(p_2)$, $y^2_t = y_t(p_2)$ (same allocation as when the active price was $p_2$).

Now, consider the same set of buyers and items as above but when the active price was $p_1$. Define a demand vector $\vec{x^1}$ such that for every buyer $i$ in $B_A(p_2)$, $x^1_i  = x_i(p_1)$, i.e., that buyer's demand when the active price was $p_1$. Let the corresponding min-cost allocation for this demand using only the items in $S_A(p_2)$ be $\vec{y^1}$. 

In summary,
\begin{itemize}
\item $(\vec{x^2}, \vec{y^2})$: For Active buyers at $p_2$, their total demand when the active price is $p_2$ and allocation using only the items $S_A(p_2)$.
\item $(\vec{x^1}, \vec{y^1}$): Same set of buyers and items as before, but with the corresponding buyer demand when the active price was $p_1$, i.e., $x_i(p_1)$.
\end{itemize}
Clearly $\vec{x^1}$ dominates $\vec{x^2}$ and therefore from Lemma~\ref{lem_diffcostmonoton}, it is clear that
$$c_t(y^1_t) \geq c_t(y^2_t) = c_t(y_t(p_2)).$$

So all we need to show now in order to prove the lemma is $c_t(y_t(p_1)) \geq c_t(y^1_t)$. 

Let the buyers in $B_A(p_1) \cap B_A(p_2)$ be $B_1$ and items in $S_A(p_1) \cap S_A(p_2)$ be $S_1$. Now, in $\vec{y^1}$, the flow from the items in $S_1$ only reaches the buyers in $B_1$ (other buyers in $B_A(p_2)$ were inactive at that price). But in the actual solution of our algorithm at active price $p_1$, the items in $S_1$ may have flow to $B_1$ and other buyers in $B_A(p_1)$. We now define two new demand vectors
\begin{itemize}
\item Let $\vec{x^1}(B_1,S_1)$ be sub-demand of $\vec{x^1}$ from the buyers in $B_1$ corresponding to their allocations from $S_1$ in $\vec{y^1}$. The allocation on $t \in S_1$ is still $y^1_t$.

\item Let $\vec{x^3}$ be defined for every $i$ as follows: $x^3_i$ is the total allocation for every buyer when the active price was $p_1$ but only from the items in $S_1$. The allocation on $t \in S^1$ is $y_t(p_1)$.
\end{itemize}

Since $\vec{x^3} \geq \vec{x^1}$ for all the common buyers, we can apply Corollary~\ref{corr_diffcostmorebuyers} and get $c_t(y_t(p_1)) \geq c_t(y^1_t).$ This completes the proof.  \end{proof}

\begin{prop_app}{prop_stopequality}
For any given item $t$ and fixed $k$, the stopping condition must be obeyed at equality. Formally, suppose that $t$ meets the stopping criterion at an active price of $p$, then
$$p - c_t(y_t(p)) = \frac{1}{k}(\lambda^{max} - c_t(y_t(p)).$$
\end{prop_app}
\begin{proof}
We first claim that for every item $t$ in the active set at an active price of $p$, $\lim_{\epsilon \to 0}c_t(y_t(p-\epsilon)) = c_t(y_t(p))$. First assume that no new item joins the active set at the price $p$. In a sufficiently small neighborhood around the price $p$, the contents of the active set cannot change. Moreover, for the restricted instance consisting only of the items and buyers in $A(p-\epsilon)$, $\vec{y}(p-\epsilon)$ is always a min-cost flow since the algorithm specifically computes a min-cost flow for the active set.

Applying Lemma~\ref{lem_contri} to the restricted instance provided by $A(p)$ where all the active items are priced at $p$, we have that $\lim_{\epsilon \to 0}c_t(y_t(p-\epsilon)) = c_t(y_t(p))$. Now, we know that for all $p-\epsilon < p$, item $t$ does not meet the active criterion and therefore
$$p-\epsilon - c_t(y_t(p-\epsilon)) < \frac{1}{k}\left(\lambda^{max} - c_t(y_t(p-\epsilon))\right).$$
Applying the limit $\epsilon \to 0$ to both sides of the above inequality, we get the desired result. Indeed, if instead of equality, we had a strict inequality, then that would imply a jump discontinuity in $c_t(y_t(p))$.

Now, what if some new item $t'$ joins the active set at price $p_t$? By definition, this would imply that, $c_{t'}(y_{t'}(p)) = p$. Then, it is clear that we still have a min-cost flow at $p$ because for any other item $t''$ in the active set, $c_{t''}(y_{t''}(p)) = p$. The same proof is therefore still applicable. \end{proof}

We now show the proof of the main bound of 1.877 on the revenue returned by our algorithm.

\begin{thm_app}{thm:1.88approx}
Algorithm~\ref{alg_uniformpeak} returns a $1.877$ approximation to the optimal profit, i.e., if $\pi^*$ is the profit of the optimal solution and $\tilde{\pi}$ is the profit of the solution returned by Algorithm~\ref{alg_uniformpeak}, then
$$\frac{\pi^*}{\tilde{\pi}} \leq 4\sqrt{e}-2-e\approx 1.877.$$
\end{thm_app}
\begin{proof}
We will now provide lower bounds on both $\pi_1$ and $\pi_2$ as returned by Algorithm~\ref{alg_uniformpeak} and show that for any instance, one of these is close enough to the optimal profit $\pi^*$. First some notation: we partition the buyers into sets $B^H$ and $B^L$ such that for the buyers in $B^H$, the minimally priced item available to them in $\vec{p^{opt}}$ is not smaller than that available to them in $\vec{p^{\sqrt{e}}}$, i.e., $i$ such that $\bar{p}^{opt}_i > \bar{p}^{\sqrt{e}}_i$. We already know due to Lemma \ref{lem_lowerboundmain} that all buyers in $B^L$ have $\bar{p}^{opt}_i > \bar{p}^{e}_i$.

Let us denote by $\vec{z^{opt}}$ the minimum-cost flow corresponding to the buyer demand $\vec{x^{opt}}$ (since the optimum allocation $\vec{y^{opt}}$ may not be a min-cost flow). Since the minimum cost flow always costs lesser than or equal to any given allocation, it is okay to compare our solutions with an upper bound on the optimum which is,
$$\pi^* \leq \sum_{i \in B}\bar{p}^{opt}_i x^{opt}_i - \sum_{t \in B}C_t(z^{opt}).$$

Now we apply the partition lemma (Lemma~\ref{lem_flowpartition}) to the envy-free profit maximizing solution with $B^H$ as the desired subset. Let $\vec{z^{opt}}(B^H)$ be the respective min-cost flow for only the demand due to buyers in $B^H$ in the profit maximizing solution. We then have,
\begin{align}
\pi^* \leq \left(\sum_{i \in B^H }\bar{p}^{opt}_i x^{opt}_i - C(\vec{z^{opt}}(B^H))\right) & +  \left(\sum_{i \in B^L} \bar{p}^{opt}_i x^{opt}_i - \left(C(\vec{z^{opt}}) - C(\vec{z^{opt}}(B^H))\right)\right)\\
\end{align}

As shown above, we have decomposed the profit into that due to $B^H$ and due to $B^L$ respectively. For convenience, we will refer to the left term above as $\pi^{opt}(B^H)$ and the right term as $\pi^{opt}(B^L)$, then we have shown that $\pi^*\leq \pi^{opt}(B^H)+\pi^{opt}(B^L)$. We now show a simple claim on the marginal costs of buyers in $B^H$ and $B^L$.
\begin{lemma}
\label{sublem_finalproofmargcost}
\begin{enumerate}
\item For all $i \in B^H$, $r_i(\vec{z^{opt}}(B^H)) \leq r_i(\vec{z^{opt}}) \leq r_i(\vec{y^{\sqrt{e}}}) \leq r_i(\vec{y^e})$.
\item For all $i \in B^L$, $r_i(\vec{z^{opt}}) \leq r_i(\vec{y^e})$.
\item $\sum_{i \in B^H}r_i(\vec{y^e})x^{opt}_i \geq C(\vec{z^{opt}}(B^H))$.
\item $\sum_{i \in B^L}r_i(\vec{y^e})x^{opt}_i \geq (C(\vec{z^{opt}}) - C(\vec{z^{opt}}(B^H)))$.
\end{enumerate}
\end{lemma}
\begin{proof}
Define $\vec{y^e}(B^H)$ and $\vec{y^{\sqrt{e}}}(B^H)$ to be the min-cost flows for the demand corresponding only to the buyers in $B^H$ in $\vec{x^e}$ and $\vec{x^{\sqrt{e}}}$ respectively. Since for $i \in B^H, \bar{p}^{opt}_i \geq \bar{p}^{\sqrt{e}}_i \geq \bar{p}^e_i$, we have $\vec{x^{opt}} \leq \vec{x^{\sqrt{e}}} \leq \vec{x^e}$. Applying Lemma~\ref{lem_diffcostmonoton} gives us the last three terms of the first result. Corollary~\ref{corr_diffcostmorebuyers} gives us that $r_i(\vec{z^{opt}}(B^H)) \leq r_i(\vec{z^{opt}})$.

For all $i$ including the buyers in $B^L$, $x_i^{opt} \leq x_i^e$ since $\vec{p^{opt}}\geq\vec{p^e}$. So applying Lemma~\ref{lem_diffcostmonoton} to the optimum solution and $\vec{x^e}$, we get the second result.

Now for the third result. Consider $\sum_{i \in B^H}r_i(\vec{y^e})x^{opt}_i$. As per the first result, this is not greater than $\sum_{i \in B^H}r_i(\vec{z^{opt}}(B^H))x^{opt}_i = c_t(z^{opt}_t(B^H))z^{opt}_t(B^H)$. By Convexity arguments, the last term is clearly an upper bound for $C(\vec{z^{opt}}(B^H))$.

For the final result, we begin with the observation (from Point 2) that
$$\sum_{i \in B^L}r_i(\vec{y^e})x^{opt}_i \geq \sum_{i \in B^L}r_i(\vec{z^{opt}})x^{opt}_i.$$
From Lemma~\ref{lem_flowpartition}, we know that $\sum_{i \in B^L}r_i(\vec{z^{opt}})x^{opt}_i \geq \sum_{t \in S}(C_t(z^{opt}_t) - C_t(z^{opt}_t(B^H))$, which completes the result. \end{proof}

\subsubsection*{Stopping Condition Related Properties}
Finally, we rewrite the stopping conditions for all buyers $i$, for $k=e$ and $k=\sqrt{e}$ respectively and show a simple lemma based on this.
$$\bar{p}^e_i = \frac{1}{e}\lambda^{max} + (1-\frac{1}{e})r_i(\vec{y^{e}}).$$
$$\bar{p}^{\sqrt{e}}_i = \frac{1}{\sqrt{e}}\lambda^{max} + (1-\frac{1}{\sqrt{e}})r_i(\vec{y^{\sqrt{e}}}).$$

\begin{lemma}
\label{sublem_relationstopcond}
For all $i$, $\bar{p}^e_i \geq \frac{1}{\sqrt{e}}\bar{p}^{\sqrt{e}}_i + (1-\frac{1}{\sqrt{e}})r_i(\vec{y^{e}}).$
\end{lemma}
\begin{proof}
Multiplying the rewritten stop condition for $k=\sqrt{e}$ with $\frac{1}{\sqrt{e}}$, we get $\frac{1}{\sqrt{e}}\bar{p}^{\sqrt{e}}_i = \frac{1}{e}\lambda^{max} + (\frac{1}{\sqrt{e}}-\frac{1}{e})r_i(\vec{y^{\sqrt{e}}}).$ From Lemma~\ref{sublem_finalproofmargcost}, we get that this quantity is no more than $\frac{1}{e}\lambda^{max} + (\frac{1}{\sqrt{e}}-\frac{1}{e})r_i(\vec{y^{e}}).$ Substituing $\frac{1}{e}(\lambda^{max} - r_i(\vec{y^e})) = \bar{p}^e_i - r_i(\vec{y^e})$ gives us the desired result.\end{proof}

\subsubsection*{Lower bound on $\pi_1$: Profit at $k=e$}

We now begin with the main proof by showing a lower bound on $\pi_1$. Recall that $\pi_1$
is the solution where all items are priced at $p^e_t$. First, we apply Lemma~\ref{lem_costcomptwoprice} with $(\vec{p^e}, \vec{x^e}, \vec{y^e})$ as the first solution and the optimum solution as the second. Following this, we apply the flow partition lemma to $\sum_{t \in S}C(z^{opt}_t)$ with $B^H$ as the subset.

\begin{align*}
\pi_1 = & \sum_{i \in B}\bar{p}^e_i x^e_i - C(\vec{y^e})\\
\geq & \sum_{i \in B}\bar{p}^e_i x^{opt}_i - C(\vec{z^{opt}})\\
= & \sum_{i \in B^H}\bar{p}^e_i x^{opt}_i - C(\vec{z^{opt}}(B^H)) + \sum_{i \in B^L}\bar{p}^e_i x^{opt}_i - (C(\vec{z^{opt}}) - C(\vec{z^{opt}}(B^H))).
\end{align*}

\textbf{$\pi_1$: Profit due to buyers in $B^H$}\\
Now consider the first term above corresponding to the buyers in $B^H$. We apply the rewritten stopping condition to get the following lower bound,
$$\sum_{i \in B^H}\bar{p}^e_i x^{opt}_i - C(\vec{z^{opt}}(B^H)) \geq \sum_{i \in B^H}\left(\frac{1}{e}\lambda^{max} + (1-\frac{1}{e})r_i(\vec{x^e})\right) x^{opt}_i - C(\vec{z^{opt}}(B^H)).$$

From Lemma~\ref{sublem_finalproofmargcost} (Point 3), we get $\sum_{i \in B^H}(1-\frac{1}{e})r_i(\vec{x^e})x^{opt}_i \geq (1-\frac{1}{e})(C(\vec{z^{opt}}(B^H)))$. And so, we get a final lower bound on the profit due to the buyers in $B^H$.
\begin{align*}
\sum_{i \in B^H}\bar{p}^e_i x^{opt}_i - C(\vec{z^{opt}}(B^H)) \geq & \sum_{i \in B^H}\frac{1}{e}\lambda^{max} x^{opt}_i - \frac{1}{e}C(\vec{z^{opt}}(B^H))\\
& \geq \frac{1}{e}\left(\sum_{i \in B^H}\bar{p}^{opt}_i x^{opt}_i - C(\vec{z^{opt}}(B^H))\right)\\
= & \frac{1}{e}\pi^{opt}(B^H).
\end{align*}

\textbf{$\pi_1$: Profit due to buyers in $B^L$}\\
We move on to the profit due to the terms in $B^L$ and apply Lemma~\ref{sublem_relationstopcond}.
\begin{equation*}
\sum_{i \in B^L}\bar{p}^e_i x^{opt}_i - (C(\vec{z^{opt}}) - C(\vec{z^{opt}}(B^H))) \geq\\
 \sum_{i \in B^L} \left(\frac{1}{\sqrt{e}}\bar{p}^{\sqrt{e}}_i + (1-\frac{1}{\sqrt{e}})r_i(\vec{y^{e}})\right)x^{opt}_i - (C(\vec{z^{opt}}) - C(\vec{z^{opt}}(B^H))).
\end{equation*}

Applying the final claim in Lemma~\ref{sublem_finalproofmargcost}, we get that $(1-\frac{1}{\sqrt{e}})r_i(\vec{y^e}) x^{opt}_i \geq (1-\frac{1}{\sqrt{e}})(C(\vec{z^{opt}}) - C(\vec{z^{opt}}(B^H))).$ Now we are in a position to get a final lower bound for the profit due to the terms in $B^L$.
\begin{align*}
= & \sum_{i \in B^L}\bar{p}^e_i x^{opt}_i - (C(\vec{z^{opt}}) - C(\vec{z^{opt}}(B^H))).\\
\geq & \frac{1}{\sqrt{e}}\left(\sum_{i \in B^L}\bar{p}^{\sqrt{e}}_i x^{opt}_i - (C(\vec{z^{opt}}) - C(\vec{z^{opt}}(B^H)))\right)\\
\geq & \frac{1}{\sqrt{e}}\left(\sum_{i \in B^L}\bar{p}^{opt}_i x^{opt}_i - (C(\vec{z^{opt}}) - C(\vec{z^{opt}}(B^H)))\right)\\
= & \frac{1}{\sqrt{e}}\pi^{opt}(B^L).
\end{align*}
Recall that $\pi^{opt}(B^H) + \pi^{opt}(B^L) \geq \pi^*$. Our final bound for $\pi_1$ reads as follows
$$\pi_1 \geq \frac{1}{e}\pi^{opt}(B^H) +  \frac{1}{\sqrt{e}}\pi^{opt}(B^L).$$

\subsubsection*{Lower bound on $\pi_2$: Profit at $k=\sqrt{e}$}
Now we move on to $\pi_2$ which is the profit due to the solution returned by our algorithm for $k=\sqrt{e}$. Notice that for the buyers in $B^H$, the demand in this solution is larger than the demand in the optimum, whereas it is smaller for the buyers in $B^L$. It may be possible that by increasing the price from the optimum solution for some buyer in $B^L$, we lose most of her flow and thus we may not be extracting any profit at all from these buyers. Our first main claim leverages a property of MHR functions to show that for every buyer in $B^L$, $x^{\sqrt{e}}_i \geq \frac{1}{2}x^{opt}_i$, i.e., due to the price increase, the drop in the buyer's demand cannot be larger than a factor of two.

Consider the function $f(x) = \lambda_i(x) - r_i(\vec{y^{\sqrt{e}}})$ for any $i \in B^L$. Recall that due to the definition of best-response demand, we have that $\lambda_i(x_i^{\sqrt{e}})=\bar{p}_i^{\sqrt{e}}$. Thus, we know from the stopping condition that $f(0) = \sqrt{e} f(x^{\sqrt{e}}_i)$. We also know from Lemma~\ref{sublem_relationstopcond} that $\sqrt{e}f(x^e_i) \geq f(x^{\sqrt{e}}_i)$. This is true because $r_i(\vec{y^e}) \geq r_i(\vec{y^{\sqrt{e}}})$ for all $i$. Therefore, applying the property of MHR functions from Lemma~\ref{lem_mhrxchange}, we get the desired claim that $x^{\sqrt{e}}_i \geq \frac{1}{2}x^e_i \geq \frac{1}{2}x^{opt}_i$. Now apply Lemma~\ref{lem_tripartition} to the solutions $(\vec{p^{\sqrt{e}}},  \vec{x^{\sqrt{e}}}, \vec{y^{\sqrt{e}}})$ and $(\vec{p^{opt}}, \vec{x^{opt}}, \vec{z^{opt}})$ with $\alpha = \frac{1}{2}$. We get the following lower bound for the profit when $k=\sqrt{e}$:

\begin{align*}
\pi_2 \geq & \left(\sum_{i \in B^H} \bar{p}_i^{\sqrt{e}} x^{opt}_i - C(\vec{z^{opt}}(B^H))\right) & + & \frac{1}{2}\left(\sum_{i \in B^L} \bar{p}^{opt}_i x^{opt}_i - (C(\vec{z^{opt}}) - C(\vec{z^{opt}}(B^H)))\right)
\end{align*}

We will refer to the above two terms as $\pi_2 (B^H)$ and $\pi_2(B^L)$.

\textbf{$\pi_2$: Profit due to buyers in $B^L$:}\\
Clearly since $\alpha = 0.5$, we immediately have that $\pi_2(B^L) = 0.5 \pi^{opt}(B^L)$. Therefore, we only need to focus on bounding $\pi_2 (B^H)$ in terms of $\pi^{opt}(B^H)$, which we do below.\\

%

\textbf{$\pi_2$: Profit due to buyers in $B^H$:}\\
So we now exclusively focus on the profit due to the terms in $B^H$. 

Using the rewritten stopping condition, we see that for all $i \in B^H$, $\bar{p}^{\sqrt{e}}_i \geq \frac{1}{\sqrt{e}}\lambda^{max} + (1-\frac{1}{\sqrt{e}})r_i(\vec{y^{\sqrt{e}}})$. But we also know that $r_i(\vec{y^{\sqrt{e}}}) \geq r_i(\vec{z^{opt}}) \geq r_i(\vec{z^{opt}}(B^H))$ due to Lemma \ref{sublem_finalproofmargcost}. Getting back to $\pi_2$, we have

\begin{align*}
\pi_2 (B^H) = & \sum_{i \in B^H}\bar{p}^{\sqrt{e}}_i x^{opt}_i - C(\vec{z^{opt}}(B^H))\\
\geq & \sum_{i \in B^H}\left(\frac{1}{\sqrt{e}}\lambda^{max} + (1-\frac{1}{\sqrt{e}})r_i(\vec{z^{opt}}(B^H))\right)x^{opt}_i -  C(\vec{z^{opt}}(B^H)).
\end{align*}

Notice that $\sum_{i \in B^H}r_i(\vec{z^{opt}}(B^H))x^{opt}_i = \sum_t c_t(z^{opt}_t(B^H))z^{opt}_t(B^H) \geq \sum_{t\in S}C_t(z^{opt}_t(B^H))$. So cancelling $(1-\frac{1}{\sqrt{e}})$ from the costs, we get
\begin{align*}
\pi_2 (B^H) \geq & \frac{1}{\sqrt{e}}(\sum_{i \in B^H}\lambda^{max}x^{opt}_i -  C(\vec{z^{opt}}(B^H)))\\
\geq & \frac{1}{\sqrt{e}}(\pi^{opt}(B^H)).
\end{align*}

Therefore, our lower bound on $\pi_2$ is $\frac{1}{\sqrt{e}}(\pi^{opt}(B^H)) + \frac{1}{2}(\pi^{opt}(B^L)).$ Recall that $\pi^{opt}(B^H) + \pi^{opt}(B^L) \geq \pi^*$ and

$$\pi_1 \geq \frac{1}{e}(\pi^{opt}(B^H)) + \frac{1}{\sqrt{e}}(\pi^{opt}(B^L)).$$

Some basic algebra gives us that
$$min(\pi_1, \pi_2) \geq \frac{1}{4\sqrt{e} - (2+e)} \pi^* \geq \frac{1}{1.877}\pi^* \approx 0.53 \pi^*.$$
This completes the proof.
\end{proof}

\subsection{Efficient Implementation of the Algorithm}
\label{app:efficient}
\begin{lem_app}{lem_complementbinarysearch}
The following invariants hold during the course of the above Algorithm.
\begin{enumerate}
\item For any $j$, $B_A(j) \cup S_A(j) = B_A(P(j)) \cup S_A(P(j))$.

\item For any $j$, all the items in $S_F(j)$ meet the stopping criterion in the interval $[P(j-1), P(j)]$ during the course of Algorithm~\ref{alg_genprocedurebody}.
\end{enumerate}
\end{lem_app}
\begin{proof}
We show this by induction on $j$. Clearly at $j =1$, no item meets the stopping criterion since since $P(1) = c_t(y^*_t)$. This means that the first invariant is also true trivially. Suppose that the invariants are true up to iteration $j-1$. Notice that at the beginning of iteration $j$, we compute a min-cost flow for the items and buyers in $B_A(j-1) \cup S_A(j-1)$. 

Assume by contradiction that at $j$,  $\exists t \in S_F(j)$ that does not meet the stopping criterion in the interval $[P(j-1), P(j)]$ in the algorithm. We carefully introduce more notation,
\begin{itemize}
\item Let $S'_F(j)$ be the subset of $S_F(j)$ of all the items that do meet the stopping criterion in the interval during the course of the algorithm.

\item Consider the set of buyers and items belonging to the active set at both $P(j-1)$ and $P(j)$, i.e., they did not meet the stopping criterion in the desired interval. Let $B^{alg}(j)$ be the set of such buyers and $S^{alg}(j)$ be the set of such items. It is not hard to see that $B^{alg}(j) \subseteq B_A(j-1)$ and $S^{alg}(j) \subseteq S_A(j-1)$. Moreover, $S^{alg}(j) \cup S'_F(j) = S_A(j-1)$.

\item Let $x^{alg}(j)$ be the demand of the buyers in $B^{alg}(j)$ at price $P(j)$ and let  $\vec{y^{alg}(j)}$ be the corresponding min-cost flow using only the items in $S^{alg}(j)$. Indeed, in algorithm~\ref{alg_genprocedurebody}, the buyers in $B^{alg}(j)$ are only using the items in $S^{alg}(j)$ at active price $P(j)$.
\end{itemize}

Now, it is not hard to see that no buyer in $B^{alg}(j)$ has an edge to any item in $S'_F(j)$ (Recall Figure~\ref{fig_sethierarchy}). Define the demand vector $\vec{x'}(j)$ as $\vec{x}(j)$ but counting only the flow sent to items in $S^{alg}(j)$. For all buyers in $B^{alg}(j)$, $x'_i(j) =  x_i(j)$. Now apply Corollary~\ref{corr_diffcostmorebuyers} for the following two instances:
\begin{enumerate}
\item $(B^{alg}(j), S^{alg}(j))$ with demand $x^{alg}(j)$ and flow $\vec{y^{alg}(j)}$
\item  $(B_A(j-1), S^{alg}(j))$ with demand $\vec{x'}(j)$ with flow being $\vec{y}(j)$ but only for the items in $S^{alg}(j)$.
\end{enumerate}
Since $t \in S^{alg}(j)$, this means that $c_t(y^{alg}_t) \leq c_t(y_t(j))$. But we already know that $t$ did not meet the stopping criteria at $P(j)$, i.e.,

$$P(j) - c_t(y^{alg}_t) < \frac{1}{k}(\lambda^{max} - c_t(y^{alg}_t)).$$

And so, $t$ could have not met the stopping criteria in $\vec{y}(j)$, which is a contradiction.

The second invariant follows almost immediately. We know that the items that reached the stopping criterion in Algorithm~\ref{alg_genprocedurebody} in $(P(j-1), P(j)]$ constitute $S_F(j)$. Let $B_F(j)$ be the corresponding buyers who also became finished along with the items in the same interval but in the original algorithm. Clearly, $B_F(j) \subseteq B_A(j-1)$ as per definition. Clearly, these buyers have edges to at least one item in $S_F(j)$. Now suppose that the binary search algorithm has found a price $p$ where item $t$ exactly meets the stopping criterion and suppose that some buyer $i \in B_F(j)$ who has an edge to this item has not been removed yet. Since we only compute min-cost flows, either this buyer is using $t$ or some item with a marginal cost equal to that of $t$, say $t'$. Indeed, $t'$'s marginal cost cannot be smaller than that of $t$ because that means that $t'$ would have been removed first. Therefore, we also remove $t'$ and along with it buyer $i$. \end{proof}

\begin{theorem}

\label{thm_efficientalgo}
The algorithm in Section~\ref{sec:efficient} returns the same solution as Algorithm~\ref{alg_genprocedurebody} for any value of $k$.
\end{theorem}
\begin{proof}
We just need to show that every item $t$ has the same price in the efficient algorithm as it does in $\vec{p^k}$. We show this inductively on the set of distinct prices in $\vec{p^k}$. Let $S_p$ be some set of items priced at $p$ in that vector and let $B_p$ be the buyers using these items. We only have to show that in the efficient algorithm, these items have the same price and these buyers are also using only the items in $S_p$. Suppose that $p \in (P(j-1), P(j)]$ for some $j$. We know that $B_p \subseteq B_A(j-1)$ and $S_p \subseteq S_A(j-1)$ as per Lemma~\ref{lem_complementbinarysearch}.

Now as per the binary search algorithm, we would be searching in an interval $(p', P(j)]$ such that for the active price strictly between $p'$ and $p$, no item would meet the stopping criterion. Also note that the buyers in $B_p$ have not been removed at this stage of the binary search because they do not have edges to any items that finished before price $p$. Therefore, we know that the binary search converges up on price $p$ and at this point we have to remove both item $B_p$ and $S_p$, items not in $S_p$ cannot meet the stopping criterion and have to have a marginal cost larger than those in $S_p$. So buyers in $B_p$ must only be using items in $S_p$ because this is a min-cost flow. \end{proof}

\section{Proofs from Section 4}
All our results in this section depend on a parameter $\Delta$ defined as follows

$$\Delta = \frac{\lambda^{max}_0}{\lambda^{min}_0}.$$

In addition, we require a slightly stronger assumption on the cost functions than just convexity. We call a production cost function $C_t(x)$ \emph{doubly convex} if its derivative $c_t(x)$ is also convex with $c_t(0)=0$. Surprisingly, without the doubly convex assumption, we show that this problem admits no good approximation algorithm. Our starting point for this result is still Algorithm~\ref{alg_genprocedurebody}. However, since $\lambda^{max}$ is no longer uniform for all functions, we redefine the stopping condition as follows.

\begin{definition}{\emph{New Stopping Criterion}($p_t, y_t, k$)\\}
\begin{align}
\label{eqn_stopconditionnew}
p_t - c_t(y_t) \geq \frac{1}{k}(\lambda^{min}_0 - c_t(y_t)).
\end{align}
\end{definition}

The above change introduces a new element into the analysis of the algorithm: an item that is inactive may already satisfy the stopping criterion if their initial price $c_t(y^*_t) = p^*_t \geq \lambda^{min}_0$. The following two statements summarize the final prices of various items according to the new stopping condition.
\begin{enumerate}
\item For any item $t$ whose initial price $p^*_t > \lambda^{min}_0$, its final price when the algorithm terminates is also $p^k_t = p^*_t$.

\item For any item $t$, whose initial price $p^*_t \leq \lambda^{min}_0$, its final price when the algorithm terminates still satisfies,
$$p^k_t - c_t(y^k_t) = \frac{1}{k}(\lambda^{min}_0 - c_t(y^k_t)).$$
\end{enumerate}

Algorithm~\ref{alg_genprocedurebody}, therefore still returns an envy-free allocation that is also a min-cost flow for the buyer demand. The analysis from Theorem~\ref{thm_ep1envyfree} can be easily extended to this case by dividing the graph into two components based on whether $p^*_t > \lambda^{min}_0$ or not. As usual denote the solution returned by Algorithm~\ref{alg_genprocedurebody} for $k=e$ as $(\vec{p^e}, \vec{x^e}, \vec{y^e})$. Now consider the definition of the following price vectors $\vec{p}(j)$ for $j=1$ to $j=\log(\Delta)$. Assume w.l.o.g that $\Delta$ is a power of $e$.

\begin{align}
\label{eqn_pricevectors}
\text{(Price Vector)}~~\vec{p}(j): & & p_t(j) = \max(p^e_t, e^{j-1}\lambda^{min}_0). 	
\end{align}

Note that we can think of $\vec{p^e}$ as $\vec{p}(0)$, since $p_t^e\geq e^{-1}\lambda^{min}_0$ due to the stopping condition.
Let $\vec{x}(j)$ be the corresponding best-response buyer demand to $\vec{p}(j)$ and $\vec{y}(j)$ be the envy-free allocation that minimizes the total cost in the space of all envy-free allocations. We later show that $\vec{y}(j)$ is actually a min-cost flow for $\vec{x}(j)$. With this definition, our algorithm becomes simple, see Algorithm \ref{alg_genmhr_bicrit}.

\begin{algorithm}[htbp]
\caption{$(O(log(\Delta), 4)$-Bicriteria approximation algorithm}
\label{alg_genmhr_bicrit}
\begin{algorithmic}[1]
\STATE Let $\pi(0)$ be the profit of $(\vec{p^e}, \vec{x^e}, \vec{y^e})$.
\STATE Let $\pi(j)$ be the profit of the envy-free solution $(\vec{p}(j), \vec{x}(j), \vec{y}(j))$.
\STATE Let $SW(0)$ be the social welfare of $(\vec{p^e}, \vec{x^e}, \vec{y^e})$.
\STATE Find the smallest $j$ such that $\pi(j) \geq \frac{1}{2}(\frac{SW(0)}{4.5(1+log(\Delta))})$ and return the corresponding solution.
\end{algorithmic}
\end{algorithm}

We show our main theorem after proving that for every $j$, $\vec{y}(j)$ is actually a min-cost flow.
\begin{claim}
For all $j$ between $1$ and $log(\Delta)$, $\vec{y}(j)$ is a min-cost flow for the demand $\vec{x}(j)$.
\end{claim}
\begin{proof}
For a given $j$, look at the prices. Items are priced at either $e^{j-1}\lambda^{min}_0$ or at $p^e_t$ if $p^e_t > e^{j-1}\lambda^{min}_0$. The solution is envy-free by definition. So divide the buyers and items as follows: let $S^H$ be the set of items with price higher than $e^{j-1}\lambda^{min}_0$ and let $B^H$ be the buyers using these items. Define $S^L$ as the items with price $e^{j-1}\lambda^{min}_0$ and $B^L$ as the corresponding buyers using these items. Now, for a solution to be a min-cost flow, all buyers should be sending flow on the items with the smallest marginal cost available to them.

By definition, we have two min-cost sub-flows: 1) buyers in $B^L$ are using the cheapest possible allocation using only the items in $S^L$, by definition; 2) the same is true for $B^H$ and $S^H$, this is because for these entities, both the prices and the allocation are exactly the same as in $\vec{p^e}, \vec{x^e}, \vec{y^e}$. So as was the case before, we only need to consider cross-edges. Moreover, since the solution is envy-free, there can be no edges going from buyers in $B^H$ to items in $S^L$. What about the reverse case, can there be a buyer $i$ in $B^L$ and an item $t$ in $S^H$ such that $c_t(y_t(j)) = c_t(y^e_t) < r_i(\vec{y}(j))$?

Consider $S^L$ and $B^L$, but in our first solution $\vec{y^e}$. Recall that these items must have a price smaller than or equal to $e^{j-1}\lambda^{min}_0$ in $\vec{p^e}$. Clearly, since our solution is both envy-free and a min-cost flow, it is clear that buyers in $B^L$ can only receive allocations from $S^L$ in our solution. Moreover, no buyer from $B^H$ has an edge to items in $S^L$. Finally, by definition, the buyers in $B^L$ have a smaller demand in $\vec{x}(j)$ as compared to $\vec{x^e}$. Therefore, applying Lemma~\ref{lem_diffcostmonoton} for the reduced sub-instance $(B^L, S^L)$, we get that $r_i(\vec{y}(j)) \leq r_i(\vec{y^e})$ for any buyer $i$ in $B^L$. But since $\vec{y^e}$ is full min-cost flow, it is true $r_i(\vec{y}(j)) \leq c_t(y^e_t)$ for any $t$ that $i$ has an edge to including items in $S^H$. This completes the proof. $\blacksquare$
\end{proof}

\begin{thm_app}{thm_generalmhr}
For any instance with MHR Demand and Doubly Convex Costs, Algorithm~\ref{alg_genmhr_bicrit} returns an envy-free solution which has a $O(\log \Delta)$-approximation to the optimal revenue and which also guarantees $\frac{1}{4}^{th}$ of the optimum welfare.
\end{thm_app}

\begin{proof}
As with our proof of Theorem~\ref{thm:1.88approx}, the approximation factor depends very crucially on non-trivial lower bounds we show for the optimal prices $\vec{p^{opt}}$.

\begin{lemma}
\label{clm_profitmaxprices2}
The price of every item $t$ in the profit-maximizing solution $\vec{p^{opt}}$ is at least its price in $\vec{p^e}$.
\end{lemma}
The proof is somewhat similar to that of Lemma~\ref{lem_lowerboundmain}, so we will only sketch the relevant details.
\begin{proof}
First, as per the new stopping condition (Equation~\ref{eqn_stopconditionnew}), there are two types of items in the solution. Those whose initial prices are larger than $\lambda^{min}_0$ (type A) and items whose prices are smaller than or equal to $\lambda^{min}_0$ (type B). Recall that for every item in type A, its price $p^e_t = p^*_t$. But we know from Lemma~\ref{lem_optisgood} that in the revenue-maximizing solution, the price of any item is at least its price in $\vec{p^*}$. This means that for all type A items, $p^{opt}_t \geq p^e_t$. So, we only need to worry about the type B items.

Assume by contradiction that in the optimal solution some items (from type B) have a price strictly smaller than their price in $\vec{p^e}$. Let $S_{min}$ be the subset of such items with the smallest price (call it $p_{min}$). As per Lemma~\ref{lem_mixedpricediffcost}, $\exists$ $t_{min} \in S_{min}$, such that $c_t(y^e_t) \leq c_t(y^{opt}_t)$. Construct the graph $G'$ as in the proof of Lemma~\ref{lem_lowerboundmain} and define the sets $S^+_{min}$ and $B^+_{min}$ accordingly. The only additional observation we need for this case is that $S^+_{min}$ cannot include any type $A$ item since $p_{min} < p^e_{t_{min}} \leq \lambda^{min}_0$, 

The rest of the proof is extremely similar to what was shown in Lemma~\ref{lem_lowerboundmain}. Since $p^e_{t_{min}} > p_{min}$, $t_{min}$ cannot satisfy the new stopping criterion ($k=e$) based on its price and allocation at $OPT$. Moreover, for every other $t \in S^+_{min}$, its price is $p_{min}$ and marginal is at least as much as that of $t_{min}$. Therefore,
\begin{equation}
\label{eqn_violatestop}
p_{min} - c_{t}(y^{opt}_{t}) < \frac{1}{e} (\lambda^{min}_0 - c_{t}(y^{opt}_t))
\end{equation}

Consider increasing the price of only the items in $S^+_{min}$ and recomputing the min-cost flow for the buyers in $B^+_{min}$ for any $p$. We define the quantities $p^+$ and $\tilde{c}(p)$ exactly as mentioned in the proof of Lemma~\ref{lem_lowerboundmain}. Let $\vec{p^+}$ be the full price vector when items in $S^+_{min}$ have a price $p^+$ and other items retain their price in $OPT$. Also, define the corresponding best-response demand $\vec{x^+}$ and envy-free allocation $\vec{y^+}$.

Our main claim is the following: the profit at $(\vec{p^+}, \vec{x^+}, \vec{y^+})$ is larger than the optimal profit which is a contradiction. The proof proceeds in the exact same manner as that of Lemma~\ref{lem_lowerboundmain}. By definition, for all $p$ smaller than $p^+$, no item meets the stopping criterion in Equation~\ref{eqn_stopconditionnew}. We remark that if some item meets the stopping condition above at price $p^+$ at all, then it must be the item(s) whose marginal cost equals $\tilde{c}(p^+)$ (See Propostion~\ref{lem_ep1stopcrithier}).

Once again, we can bound the difference in profits as with Lemma~\ref{lem_lowerboundmain}, and make the claim that for all $i$, $(\lambda_i(x^+_i) - \tilde{c}(p^+))x^+_i - (\lambda_i(x^{opt}_i) - \tilde{c}(p^+))x^{opt}_i > 0$. Define $f_i(x) = \lambda_i(x) - \tilde{c}(p^+)$. The only property required to show the claim that we make is $f(0) \geq  ef_i(x^+_i)$. But, from the stopping criterion, we know that

$$p^+ - \tilde{c}(p^+) \leq \frac{1}{e}(\lambda^{min}_0 - \tilde{c}(p^+)) \leq \frac{1}{e}(\lambda_i(0) - \tilde{c}(p^+)).$$

Indeed, the above inequalities are true because for all $i$, $\lambda_i(0) \geq \lambda^{min}_0$. The rest of the proof follows. $\blacksquare$ \end{proof}

We are now ready to prove our main theorem, that the solution returned by Algorithm~\ref{alg_genmhr_bicrit} gives us a good fraction of both revenue and welfare. We know that in $\vec{p^{opt}}$, every item is priced between its price in $\vec{p^e}$ and $\lambda^{max}_0$, i.e., the latter being the largest valuation any infinitesimal buyer may hold for the items. Next, define $\vec{z^{opt}}$ to be the min-cost flow for the buyer demand $\vec{x^{opt}}$ since $\vec{y^{opt}}$ is envy-free but not necessarily cost minimizing. 

Finally every $j$, define $SW(j)$ to be the social welfare of the solution $(\vec{p}(j), \vec{x}(j), \vec{y}(j))$.  Let $\pi^*$ be the optimal profit and $SW^*$ be the welfare of the social welfare maximizing solution. We show our result using two small lemmas, which we state first and then prove after showing how this leads to the main result.
\begin{lemma}
\label{sublem_boundonsw0}
$SW(0) \geq \pi^*$ (the optimal profit) and $SW(0) \geq  \frac{1}{2}SW^*.$
\end{lemma}

\begin{lemma}
\label{sublem_welfaredifference}
For all $j$, $SW(j) - SW(j+1) \leq 4.5\cdot\pi(j)$.
\end{lemma}

We show how these lemmas lead to the main theorem and then prove the actual lemmas. First, we show that Algorithm~\ref{alg_genmhr_bicrit} must return at least one such solution which satisfies the desired lower bound on the profit.
Summing up Lemma~\ref{sublem_welfaredifference} from $j=0$ to $j=1+\log(\Delta)$, we get that $SW(0) \leq 4.5 \sum_j \pi(j)$. If the algorithm does not return even one such $j$, then it means that for every $j$,
\begin{equation}
\label{eqn_failedj}
\pi(j) < \frac{1}{2 \times 4.5} \frac{SW(0)}{1+\log(\Delta)}.
\end{equation}

Summing up, we get $SW(0) \leq \frac{1}{2}SW(0)$, which is not true since then $\pi(0)$ would be returned by the algorithm. Thus, the algorithm must return some solution. Suppose that the index $j$ returned by the algorithm is $j^*$. Then, by definition every $j < j^*$ must satisfy Equation(15).

The revenue bound in the bicriteria result is trivial to see because $SW(0) \geq \pi^*$ and our solution satisfies

$$\pi(j^*) \geq \frac{1}{2} \frac{SW(0)}{4.5(1+\log(\Delta))}.$$

Now, consider the following quantity,
\begin{align*}
SW(0) - SW(j^*) &=& \sum_{j = 1}^{j^*} SW(j-1) - SW(j) & \\
&\leq& \sum_{j=0}^{j^*-1} 4.5\cdot \pi(j) && \text{(Lemma~\ref{sublem_welfaredifference}))} \\
&<& 4.5 \sum_{j=0}^{j^*-1} \frac{1}{2} \frac{SW(0)}{4.5(1+\log(\Delta))} & \\
& \leq& \frac{1}{2} SW(0) && \text{($j^* \leq 1+log(\Delta)$)}.
\end{align*}

So, $SW(0) - SW(j^*) \leq \frac{1}{2}SW(0)$, which implies that our solution's social welfare $SW(j^*)$ is at least half of $SW(0)$ which is one-fourth of the optimal welfare by Lemma~\ref{sublem_boundonsw0}. We now prove the small lemmas. \\
\textbf{(Proof of Lemma~\ref{sublem_boundonsw0})}\\
We know that in the revenue-maximizing solution, every buyer's demand is smaller than in $\vec{x^e}$. Also recall that $\vec{z^{opt}}$ is the min-cost flow for the demand $\vec{x^{opt}}$. Let $SW(\vec{x^{opt}}, \vec{y^{opt}})$ be the welfare of the profit maximizing solution. It is not hard to see that $SW(\vec{x^{opt}}, \vec{y^{opt}}) \geq \pi^*$, because the value of every infinitesimal buyer is at least the price that she is paying. The costs are the same in both cases. Next, recall that,

$$SW(\vec{x^{opt}}, \vec{y^{opt}}) \leq \sum_{i} \int_{x=0}^{x^{opt}_i}\lambda_i(x)dx - C(\vec{z^{opt}}).$$

We can write $SW(0)$ as follows (recall that $\lambda_i$ is decreasing as $x$ increases and so the minimum value in the interval $x^{opt}_i$ to $x^e_i$ is $\lambda_i(x^e_i)$, which in turn is the payment by the buyer).
\begin{align*}
SW(0) = & \sum_i \int_{x=0}^{x^e_i}\lambda_i(x)dx - C(\vec{y^e})\\
\geq & SW(\vec{x^{opt}}, \vec{y^{opt}}) +  \sum_i \int_{x=x^{opt}_i}^{x^e_i}\lambda_i(x)dx - (C(\vec{y^e}) - C(\vec{z^{opt}}))\\
\geq & SW(\vec{x^{opt}}, \vec{y^{opt}}) + \sum_i \lambda_i(x^{e}_i)(x^e_i - x^{opt}_i) - (C(\vec{y^e}) - C(\vec{z^{opt}}))\\
= & SW(\vec{x^{opt}}, \vec{y^{opt}}) + \sum_i \bar{p}^e_i (x^e_i - x^{opt}_i) - (C(\vec{y^e}) - C(\vec{z^{opt}}))
\end{align*}

Now, create a new instance with buyer set $B \cup B'$ such that there is one buyer $i'$ in $B'$ for each buyer $i$ in $B$ with access to the same set of items. Consider the demand vector $\vec{x^2}$ such that for all $i \in B$, $x^2_i = x^{opt}_i$ and for $i' \in B'$, $x^2_{i'} = x^e_i - x^{opt}_i$. Let the corresponding min-cost flow be $\vec{y^2}$. Clearly $\vec{y^2}$ and $\vec{y^e}$ must have the exact same total allocation on all items. On the other hand, notice that the min-cost flow of just the demands of buyers in $B$ is exactly $\vec{z^{opt}}$. Therefore, we can apply the flow partition lemma (Lemma \ref{lem_flowpartition}) for buyers $B$ and $B'$, and obtain that $C(\vec{y^e}) - C(\vec{z^{opt}})\leq \sum_{i' \in B'} r_{i'}(\vec{y^2})(x^e_i - x^{opt}_i)$.

Moreover, we know that $\vec{p^e} \geq \vec{p^*}$, and so for all $i$, $r_i(\vec{y^e}) \leq \bar{p}^*_i \leq \bar{p}^e_i$. Therefore, for all $i'$, $r_{i'}(\vec{y^2}) = r_i(\vec{y^e}) \leq \bar{p}^e_{i}$. Therefore we know that the term $C(\vec{y^e}) - C(\vec{z^{opt}})$ is at most $\sum_{i' \in B'} r_{i'}(\vec{y^2})(x^e_i - x^{opt}_i) = \sum_i r_i(\vec{y^e})(x^e_i - x^{opt}_i) \leq \bar{p}^e_i(x^e_i - x^{opt}_i)$. Therefore, the second term in the lower bound for $SW(0)$ above is non-negative and we can bound $SW(0)$ as

$$SW(0) \geq SW(\vec{x^{opt}}, \vec{y^{opt}}) \geq \pi^*,$$
as desired.

Next, we need to show that $SW(0) \geq \frac{1}{2} SW^*.$ This proof is exactly identical to the bicriteria result we showed in Theorem~\ref{thm:welfare}. Notice that the only requirement for the theorem was that for all $i$,
\begin{equation}
\label{eqn_neededforbicrit}
\frac{\lambda_i(x^e_i) - r_i(\vec{y^e})}{|\lambda'_i(x^e_i)|} \leq x^e_i.
\end{equation}

To show that this still holds, consider $f_i(x) = \lambda_i(x) - r_i(\vec{y^e})$. This function is MHR and at $x=x^e_i$ satisfies $f_i(0) \geq e f_i(x^e_i)$ by the stopping condition. So, by Lemma~\ref{lem_subclaim_mhr}, Equation~\ref{eqn_neededforbicrit} is valid here as well. Notice that for any $i$ whose price is larger than $\lambda^{min}_0$, $f_i(x^e_i) =  \lambda_i(x^e_i) - r_i(\vec{y^e}) = 0$ since all of its items were inactive throughout the runtime of the algorithm. $\qed$ \\

\textbf{(Proof of Lemma~\ref{sublem_welfaredifference})}: $SW(j) - SW(j+1) \leq 4.5\pi(j)$ \\
Recall that for all $i$, $x_i(j+1)$ is no larger than $x_i(j)$.
\begin{align*}
SW(j) - SW(j+1) = & \sum_i \int_{x_i(j+1)}^{x_i(j)} \lambda_i(x)dx - (C(\vec{y}(j) - C(\vec{y}(j+1)) \\
\leq & \sum_i \lambda_i(x_i(j+1)) (x_i(j) - x_i(j+1)) - (C(\vec{y}(j)) - C(\vec{y}(j+1))\\
\leq & \sum_i \lambda_i(x_i(j+1)) x_i(j) - C(\vec{y}(j)).
\end{align*}

The last inequality is true because it just states that the profit made at $\pi(j+1)$ is non-negative, i.e.,
$$\sum_i \lambda_i(x_i(j+1))x_i(j+1) =\sum_i \bar{p}_i(j+1)x_i(j+1) \geq C(\vec{y}(j+1)).$$

Now, look at the profit $\pi(j)$.
$$\pi(j) = \sum_i \bar{p}_i(j)x_i(j) - C(\vec{y}(j)).$$

We claim that for every $i$, $\bar{p}_i(j)$ is at least $\frac{1}{e}\lambda_i(x_i(j+1)) = \frac{1}{e} \bar{p}_i(j+1).$ This is true by definition of the price vectors, because we are increasing by at most a factor $e$. By definition, $\bar{p}_i(j+1) = \max(\bar{p}^e_i, e^j \lambda^{min}_0).$ If $\bar{p}_i(j+1) = \bar{p}^e_i$, then the price at $j^{th}$ iteration is also the same and so the $\frac{1}{e}$ factor is trivially true. Otherwise, $\bar{p}_i(j) = \max(\bar{p}^e_i, e^{j-1} \lambda^{min}_0) \geq e^{j-1} \lambda^{min}_0 = \frac{1}{e}\bar{p}^e_i$. Therefore, we are now ready to complete our bound.

\begin{align*}
\frac{SW(j) - SW(j+1)}{\pi(j)} \leq & \frac{\sum_i \bar{p}_i(j+1) x_i(j) - C(\vec{y}(j))}{\sum_i \bar{p}_i(j)x_i(j) - C(\vec{y}(j))}\\
\leq & \frac{\sum_i e\bar{p}_i(j) x_i(j) - C(\vec{y}(j))}{\sum_i \bar{p}_i(j)x_i(j) - C(\vec{y}(j))}.
\end{align*}

Our next and final claim is that $C(\vec{y}(j)) \leq 0.5 (\sum_{i \in B}\bar{p}_i(j) x_i(j))$. For any doubly convex function, $C_t(x) \leq 0.5x\cdot c_t(x)$. Moreover, we know that since $\vec{x}(j) \leq \vec{x^*}$, $r_i(\vec{y}(j)) \leq r_i(\vec{y^*}) \leq \bar{p}_i(j)$. All these identities simply come from the fact that reducing demand can only lead to a reduction in marginal cost. Therefore, we have
\begin{align*}
 C(\vec{y}(j)) \leq & 0.5 \sum_t c_t(y_t(j))y_t(j) \\
 & = 0.5 \sum_{i \in B} r_i(\vec{y}(j)) x_i(j) \\
 & \leq 0.5 \sum_{i \in B} \bar{p}_i(j) x_i(j)
\end{align*}

Consider the ratio between the difference in welfare to the profit:
$$\frac{SW(j) - SW(j+1)}{\pi(j)} \leq \frac{\sum_i e\bar{p}_i(j) x_i(j) - C(\vec{y}(j))}{\sum_i \bar{p}_i(j)x_i(j) - C(\vec{y}(j))}.$$

It is not hard to see that the RHS is largest when $C(\vec{y}(j)) = 0.5 \sum_i \bar{p}_i(j_) x_i(j)$, 
and so $\frac{SW(j) - SW(j+1)}{\pi(j)} \leq 2e-1 \leq 4.45$. \end{proof}

\subsection{Complexity Results}

The following complexity results hold for our problem without the uniform peak assumption.

\begin{prop_app}{clm_complexitygeneralmhr}
\begin{enumerate}
\item There cannot be a $O(\Delta^{k})$-approximation algorithm for any $k > 0$ for UDP in Large Markets with MHR Inverse Demand and Convex Costs (instead of doubly convex) unless $NP \subseteq DTIME(n^{(log^{c}n)})$ for some constant $c$.

\item There is no constant factor approximation algorithm for our UDP problem in large markets with MHR inverse demand and doubly convex costs unless $NP \subseteq DTIME(n^{(log^{c}n)})$ for some constant $c$.
\end{enumerate}
\end{prop_app}
\begin{proof}
The reductions are from $UDP$-$UV$, i.e., the unit demand pricing problem in small markets with uniform valuations and unlimited supply. It was shown in~\cite{chalermsook2012improved} that this problem is $\log^{1-\epsilon}(|B|)$-hard to approximate for any constant $\epsilon > 0$, unless $NP \subseteq DTIME(n^{(log^{\epsilon'}n)})$, where $\epsilon'$ is some constant that depends only on $\epsilon$. Fixing some value of $\epsilon > 0$, this implies that $UDP$-$UV$ cannot have a constant factor approximation algorithm unless $NP \subseteq DTIME(n^{(log^{c}n)})$, where $c$ is now a constant (since we have fixed some $\epsilon > 0$). Therefore, we will show both our hardness results above by proving that the statement would indicate a constant approximation algorithm for $UDP$-$UV$.

\noindent \textbf{\textbf{(Statement 1):}}\\
Assume by contradiction that $\exists$ a $O(\Delta^k)$ approximation algorithm for some constant $k > 0$. Then, we show that for UDP with uniform valuations and unlimited supply, there must exist a $2^k$ constant factor approximation algorithm. Consider some instance of UDP-Uniform valuations with buyer set $B'$ and items $S'$. Suppose that the maximum buyer value is $v_{max}$ and minimum buyer value is $v_{min} > 0$. Consider the following linear transformation on the buyer values,
$$\lambda_i(x) = 1 + \frac{v_i}{v_{max}} \text{ for } x \leq 1.$$

Define an instance of our problem with one buyer type $i$ for every buyer $i \in B'$. The inverse demand function is uniform as defined above. Moreover, the item set coincides with $S$ but now having the cost function $C_t(x) = x$ for every single item. Look at the demand functions: $\lambda^{max}_0  = 2$ and $\lambda^{min}_0 \geq 1$. Therefore, $\Delta = 2$ and we can always obtain a $2^k$-approximate solution for this instance. We show that this gives a $2^k$-approximate solution for the original UDP problem as well.

Consider any solution $\vec{p}, \vec{x}, \vec{y}$ of our problem. Assume w.l.o.g (see proposition~\ref{prop_mhrnottoobad}) that every buyer receives an integral amount of one single good, i.e, $y_t(i) = 0$ or $1$ for every $i$. Consider the following price vector $\vec{p^2}$ defined as

$$p^2_t = v_{max} (p_t - 1).$$

We claim that the revenue of $(\vec{p^2}, \vec{x}, \vec{y})$ for the original $UDP$-$UV$ problem is exactly the revenue of our solution scaled by a factor $v_{max}$. First, we show that the allocation is envy-free for these new prices. It is not hard to see that the transformation is linear and so the order of prices is maintained for the items, which means that every buyer's cheapest item set remains the same. Next, suppose buyer $i$ is sending flow on item $t$ for our solution. Then,
$$p_t \leq 1+\frac{v_i}{v_{max}} \implies p^2_t \leq v_i.$$

So the flow is definitely feasible. The profit of our solution is $\sum_{t} (p_t y_t - C_t(y_t)) = \sum_t ((p_t -1)y_t = \sum_t \frac{1}{v_{max}}(p^2_t y_t)$.

So this is just a scaled down version of the profit of the UDP problem. Therefore, every solution of our problem can be converted to a solution of UDP with scaled profits. It is not hard to see that the same is applicable for the reverse direction as well and so the optima must coincide. So, any $2^k$-Approximate solution for our problem must retain this factor for the original UDP as well. $\blacksquare$

\noindent \textbf{(Statement 2):}\\
This follows directly from our reduction in Propostion~\ref{prop_mhrnottoobad} as a constant factor approximation for our problem would indicate a constant factor approximation for $UDP$-$UV$ as well. \end{proof}

\section{Some more properties of minimum cost flows}\label{sec:min-cost}

\begin{lemma}
\label{lem_costdiff}
Consider two buyer demands $\vec{x^1}$ and $\vec{x^2}$ where $\vec{x^2} \geq \vec{x^1}$ and the corresponding min-cost flows are $\vec{z^1}$ and $\vec{z^2}$. Then, the following provides a bound for the difference in costs
$$\sum_t (C_t(z^2_t) - C_t(z^1_t)) \geq \sum_{i} r_i(\vec{z^1}) (x^2_i - x^1_i).$$
\end{lemma}
\begin{proof}
The proof follows from Corollary~\ref{corr_flowmagnitudemincost}. We know that for $t$, either $z^2_t \geq z^1_t$ or $c_t(z^2_t) = c_t(z^1_t)$. Therefore, we can write for all $t$, $C_t(z^2_t) - C_t(z^1_t) \geq c_t(z^1_t) (z^2_t - z^1_t)$. So we have,

$$\sum_t (C_t(z^2_t) - C_t(z^1_t)) \geq \sum_t c_t(z^1_t) z^2_t - \sum_t c_t(z^1_t) z^1_t = \sum_t c_t(z^1_t) z^2_t - \sum_{i \in B}r_i(\vec{z^1})x^1_i.$$

The last equation is simply a rearrangement of flow from the items to the buyers. Now, consider any buyer $i$ and item $t$ such that $z^2_t(i) > 0$. We know that $c_t(z^1_t) \geq r_i(\vec{z^1})$ because the latter represents the minimum marginal cost of any item $i$ has access to. This gives us $\sum_t c_t(z^1_t) z^2_t \geq \sum_{i \in B}r_i(\vec{z^1})x^2_i$, which completes the proof.  \end{proof}

\begin{lemma}
\label{lem_mixedpricediffcost}
Consider a price vector $\vec{p^1}$ and a corresponding solution $(\vec{x^1}, \vec{y^1})$ which is a best-response envy-free allocation such that $\vec{y^1}$ is also a min-cost flow to the demand $\vec{x^1}$. Let $\vec{p^2}$ be another price vector such that there is at least one item whose price in $\vec{p^2}$ is smaller than its price in $\vec{p^1}$. Let $(\vec{x^2},\vec{y^2})$ be an envy-free allocation for $\vec{p^2}$. Then, $\exists t \in S$, such that $p^1_t > p^2_t$ and $c_t(y^1_t) \leq c_t(y^2_t)$.
\end{lemma}
Note that $\vec{y^2}$ need not be a min-cost flow for the demand $\vec{x^2}$.
\begin{proof}
What the lemma says is that at least one of the items whose price in $\vec{p^2}$ is smaller than the price in $\vec{p^1}$ has a marginal cost that is larger or equal in the B-R, E-F allocation corresponding to $\vec{p^2}$. Intuitively this is not hard to see because reducing the price leads to an increase in the allocation.

Consider the set of all items whose price in $\vec{p^2}$ is strictly smaller than its price in $\vec{p^1}$ and let $S_{min}$ be the subset of such items with the smallest price, $S_{min} = \left\lbrace t | t \in \argmin_{p^{1}_{t_1} > p^2_{t_1}}(p^2_{t_1}) \right\rbrace$. Let $B^1_{min}$ be the set of buyers who receive non-zero amounts of the items in $S_{min}$ in $\vec{y^1}$. Finally, let $p_{min}$ be the price of all the items in $S_{min}$.

We claim that for any buyer $i$ in $B^1_{min}$, their entire allocation in $\vec{y^2}$ can only come from the items in $S_{min}$. If this is not true, the only other possible case is that the buyer is receiving allocation from some other item $t$ not in $S_{min}$ but still priced at $p_{min}$. But by definition of $S_{min}$, $t$'s price in $\vec{p^1}$ is not larger than $p_{min}$ which means $i$ cannot have an edge to such a $t$ (or else $\vec{y^1}$ could cease to be envy-free). 
%

Now construct the following reduced demand vectors $\vec{x^{R_1}}$ and $\vec{x^{R_2}}$ for instances with the full buyer set but the item set is reduced to $S_{min}$. For buyer $i$, $x^{R_1}_i$ is the total amount that buyer $i$ receives only from the items in $S_{min}$. This is zero from some buyers and non-zero for buyers in $B^1_{min}$ by definition. Let $\vec{y^{R_1}}$ be the exact same allocation as $\vec{y^1}$ but only for the items in $S_{min}$. Since $\vec{y^1}$ is a min-cost flow, the sub-flow on the items in $S_{min}$ for the given demand must also be one of minimum cost.

Similarly define $\vec{x^{R_2}}$ as the flow sent by buyers only to the items in $S_{min}$ and $\vec{y^{R_2}}$ as the corresponding allocation on these items. Since we have shown that the the buyers in $B^1_{min}$ use these items exclusively, this means that $\forall i \in B^1_{min}$, $x^{R_2}_i = x^2_i \geq x^1_i \geq x^{R_1}_i$. Moreover, since all the items in $S_{min}$ are priced at $p_{min}$ and the allocation is a min-cost, E-F allocation, it means that for items at the same price, we have a min-cost flow.

Now, using $(\vec{x^{R_1}}, \vec{y^{R_1}})$ and $(\vec{x^{R_2}}, \vec{y^{R_2}})$ as inputs to Lemma~\ref{lem_diffcostmonoton}, we get that for all $t \in S_{min}$, $c_t(y^1_t) \leq c_t(y^2_t)$. \end{proof}

\begin{lemma}
\label{lem_costcomptwoprice}
Consider two price vectors $\vec{p^1}$ and $\vec{p^2}$ where $\vec{p^2} \geq \vec{p^1}$, and their corresponding best-response buyer demands $\vec{x^1}$ and $\vec{x^2}$. Suppose that $\vec{z^1}$ and $\vec{z^2}$ are the corresponding min-cost flows for these demand vectors. As long as for all $t$, we have that $p_t^1\geq c_t(z_t^1)$ and $p_t^2\geq c_t(z_t^2)$, then it must be true that:

$$\sum_{i \in B}\bar{p}^1_i x^1_i - \sum_{t \in S}C_t(z^1_t) \geq \sum_{i \in B}\bar{p}^1_i x^2_i - \sum_{t \in S}C_t(z^2_t).$$
\end{lemma}
\begin{proof}
Note that these allocations may not be envy-free. The left part of this inequality is exactly the profit of the solution, but the right parts uses prices $\vec{p^1}$ with the allocation and demand $\vec{x^2}$, $\vec{z^2}$.

Remember for all $i \in B$, $x^1_i \geq x^2_i$. Applying Lemma~\ref{lem_diffcostmonoton} and Corollary~\ref{corr_flowmagnitudemincost} for these two demand vectors, we can say that for any item $t$, either $z^1_t \geq z^2_t$ or $c_t(z^1_t) = c_t(z^2_t)$. Therefore, in both these cases, it is easy to show that $C_t(z^1_t) \leq C_t(z^2_t) + c_t(z^1_t)(z^1_t - z^2_t)$. Using this in the LHS of the main lemma inequality, we get

\begin{equation}
\label{eqn_intermed1}
\sum_{i \in B}\bar{p}^1_i x^1_i - \sum_{t \in S}C_t(z^1_t) \geq \sum_{i \in B}\bar{p}^1_i x^2_i - \sum_{t \in S}C_t(z^2_t) + \left(\sum_{i \in B}\bar{p}^1_i (x^1_i - x^2_i) - \sum_{t \in S}c_t(z^1_t)(z^1_t - z^2_t) \right).
\end{equation}

If we can show that the second term in the RHS that has been parenthesized is non-negative, we are done. Recall that for every $i$, $r_i(\vec{x^1})$ is the marginal cost of any item being used by $i$ in $\vec{z^1}$. We also know that $\overline{p^1_i} \geq r_i(\vec{x^1})$ by our assumption. Consider the quantity $\sum_{t \in S}c_t(z^1_t)(z^1_t - z^2_t)$.

$$\sum_{t \in S}c_t(z^1_t)z^1_t - \sum_{t \in S}c_t(z^1_t)z^2_t = \sum_{i \in B}r_i(\vec{x^1})x^1_i - \sum_{t \in S}c_t(z^1_t)z^2_t.$$
The above equation is simply via a rearrangement of the allocation from the items to the buyers. Now, we claim that $\sum_{t \in S}c_t(z^1_t)z^2_t \geq \sum_{i \in B}r_i(\vec{x^1})x^2_i$. Look at the all the items $t$ where user $i$'s allocation $x^2_i$ turns up in the vector $\vec{z^2}$. We claim that for all such items, $c_t(z^1_t) \geq r_i(\vec{z^1})$. This has to hold because $\vec{z^1}$ is a min-cost flow and therefore for any other item $t$ that $i$ has access to, its marginal cost has to be greater than or equal to $r_i(\vec{z^1})$.

Therefore using this in Equation~\ref{eqn_intermed1},
\begin{align*}
\sum_{i \in B}\bar{p}^1_i (x^1_i-x^2_i) - \sum_{t \in S}c_t(z^1_t)(z^1_t - z^2_t) \geq & \sum_{i \in B}\bar{p}^1_i (x^1_i - x^2_i) - \sum_{i \in B}r_i(\vec{z^1})(x^1_i-x^2_i)\\
= & \sum_{i \in B}(\bar{p}^1_i - r_i(\vec{z^1}))(x^1_i-x^2_i)\\
\geq & 0.
\end{align*}
The above term is strictly non-negative because because $\bar{p}_i^1 \geq r_i(\vec{z^1})$ and $x^1_i \geq x^2_i$ for all $i$. This completes the proof.
$\blacksquare$ \end{proof}

\begin{lemma}
\label{lem_contri}
Let $G$ be a bipartite market and let $\vec{x}(p)$ denote the corresponding best-response buyer demand vector when all the items in the market have a price of $p$. Let $\vec{z}(p)$ denote the min-cost flow at this demand. Then, for any given $\bar{p}$ and item $t$, its marginal cost is left continuous at $\bar{p}$. Formally,
$$\lim_{\epsilon \to 0}c_t(z_t(\bar{p}-\epsilon)) = c_t(z_t(\bar{p})).$$

\end{lemma}
\begin{proof}
Assume by contradiction that for some $t$, $\lim_{\epsilon \to 0}c_t(z_t(\bar{p}-\epsilon)) \neq c_t(z_t(\bar{p}))$. We have already shown that as the price increases, the marginal cost cannot increase in Lemma~\ref{lem_diffcostmonoton}. Therefore, this can only mean that $\lim_{\epsilon \to 0}c_t(z_t(\bar{p}-\epsilon)) > c_t(z_t(\bar{p}))$. The monotonicity of the marginal cost also means that for every other item $t'$, $\lim_{\epsilon \to 0}c_{t'}(z_{t'}(\bar{p}-\epsilon)) \geq c_{t'}(z_{t'}(\bar{p}))$

Consider any user $i$, it is not hard to reason out that $x_i(\bar{p})$ must be continuous. Formally, $\lim_{\epsilon \to 0}x_i(\bar{p}-\epsilon) = x_i(\bar{p})$. This holds because $\lambda_i(x)$ is continuous and differentiable. So, we have the following,

$$\sum_{t'} z_{t'}(\bar{p}) = \sum_i x_i(\bar{p}) = \sum_i \lim_{\epsilon \to 0}x_i(\bar{p}-\epsilon) = \lim_{\epsilon \to 0} \sum_i x_i(\bar{p}-\epsilon) = \lim_{\epsilon \to 0} \sum_t z_t(\bar{p}-\epsilon).$$

Now the claim is that the last term in the right hand side of the above equation is strictly larger than $\sum_{t'} z_{t'}(\bar{p})$. This is true because for every $t'$, the limit of $z_{t'}(\bar{p}-\epsilon)$ is greater than or equal to $z_{t'}(\bar{p})$ and for item $t$, this inequality happens to be strict. This is a contradiction and therefore, the marginal cost at the item has to be continuous.
$\blacksquare$ \end{proof}

\begin{lemma}
\label{lem_tripartition}
Let $(\vec{p}, \vec{x}, \vec{z})$ be a given solution where $\vec{x}$ is a best-response demand to $\vec{p}$ and $\vec{z}$ is the corresponding min-cost flow. Let $(\vec{p^b}, \vec{x^b}, \vec{z^b})$ be some benchmark solution for the same instance. Suppose that $B^H$ is the set of buyers to whom the minimum priced item available in $\vec{p^b}$ is not smaller than that in $\vec{p}$ and $B^L = B \setminus B^H$. Moreover, there is a fraction $\alpha \leq 1$ such that for all buyers in $B^L$, $x_i \geq \alpha x^b_i$. Then, consider the following quantity,
$$\pi = \sum_{i \in B}\bar{p}_i x_i - C(\vec{z}).$$  If we denote by $\vec{z^b}(B^H)$ the min-cost flow for the demand $\vec{x^b}$ but only for the buyers from $B^H$, then

$$\pi \geq \left(\sum_{i \in B^H} \bar{p}_i x^b_i - C(\vec{z^b}(B^H))\right) + \alpha\left(\sum_{i \in B^L} \bar{p}^b_i x^b_i - (C_t(\vec{z^b}) - C_t(\vec{z^b}(B^H)))\right).$$
\end{lemma}
Recall that we are only looking at prices obeying $\vec{p}, \vec{p^b} \geq \vec{p^*}$ (Assumption~\ref{assumption_weakpricelowerbound}). The above term deserves careful attention. For the buyers in $B^H$, the price is from $\vec{p}$ but the demand and allocation are from the benchmark solution. For the other set of buyers, both the price and the allocation are from the benchmark. Also notice that $\vec{z}$ or $\vec{z^b}$ may not be envy-free and therefore, $\pi$ is the profit that the seller makes at the cheapest (non-envy-free solution), given the price and the demand. 
\begin{proof}
Notice that for the buyers in $B^H$, $\vec{x}$ dominates $\vec{x^b}$ and for the buyers in $B^L$, $\vec{x}$ dominates $\alpha\vec{x^b}$. Therefore, in order to get an allocation only in terms of $\vec{x^b}$ for the buyers in $B^H$ and $\alpha \vec{x^b}$ for buyers in $B^L$, we need to show that the profit due to the additional demand from $\vec{x}-\vec{x^b}$ and $\vec{x}-\alpha \vec{x^b}$ for the respective buyer sets is non-zero. This can be done via the flow partition lemma by creating dummy buyers and reallocating this flow to them. 

Construct a new instance with a set of buyers $B_1 \cup B_2$, such that for every $i \in B$, there is an $i_1$ in $B_1$ and $i_2$ in $B_2$ with access to the same set of items as $i$. Now construct a new demand vector $\vec{x'}$ such that for all $i \in B^H$, $x'_{i_1} = x^b_i$ and for the corresponding $i_2 \in B_2$, $x'_{i_2} = x_i - x^b_i$. Similarly $\forall i \in B^L$, $x'_{i_1} = \alpha x^b_i$ and for the corresponding $i_2 \in B_2$, $x'_{i_2} = x_i - \alpha x^b_i$. The min-cost flow for $\vec{x'}$ must coincide with $\vec{z}$. Moreover, look at the demand from the buyers in $B_1$, i.e., $\vec{x'}(B_1)$, let the corresponding min-cost for this demand alone be $\vec{z}(B_1)$.

Applying the flow partition lemma for this new instance and with $B_1$ being the desired subset of the total buyer set. We get,
$$C(\vec{z}) = C(\vec{z}(B_1)) + (C(\vec{z}) - C(\vec{z}(B_1))) \leq C(\vec{z}(B_1)) + \sum_{i_2 \in B_2}r_{i}(\vec{z})x'_{i_2}.$$

The last term above comes from the fact that the marginal cost faced by every $i_2$ must be equal to the marginal cost faced by the corresponding $i$ in $\vec{z}$. We slightly abuse notation here since we sum over $i_2$, but the marginal cost is that faced by $i$ in $\vec{z}$, where $i_2$ is the corresponding buyer to that $i$ in $B_2$. Moreover, since $\vec{x^*}$ dominates $\vec{x}$, by the weak price lower bound condition (Assumption~\ref{assumption_weakpricelowerbound}), this implies $r_i(\vec{z}) \leq \bar{p}_i$, giving us,

$$C(\vec{z}) - C(\vec{z}(B_1)) \leq \sum_{i_2 \in B_2}r_{i}(\vec{z})x'_{i_2} \leq \sum_{i_2 \in B_2}\bar{p}_{i} x'_{i_2}.$$

Now dividing up the term $\pi$ into the two buyer sets, we get
\begin{align*}
\pi = & \sum_{i_1 \in B_1}\bar{p}_{i} x'_{i_1} - C(\vec{z}(B_1)) + \sum_{i_2 \in B_2}\bar{p}_{i} x'_{i_2} - (C(\vec{z}) - C(\vec{z}(B_1))\\
\geq & \sum_{i_1 \in B_1}\bar{p}_{i} x'_{i_1} - C(\vec{z}(B_1)) + \sum_{i_2 \in B_2}\bar{p}_{i} x'_{i_2} - \sum_{i_2 \in B_2}\bar{p}_{i} x'_{i_2}\\
\geq & \sum_{i_1 \in B_1}\bar{p}_{i} x'_{i_1} - C(\vec{z}(B_1))
\end{align*}

Now consider the allocation $\vec{z}(B_1)$. $B_1$ is merely a repetition of the buyer set $B$ but with a different flow on the buyers, namely $x^b_i$ on the buyers corresponding to $B^H$ and $\alpha x^b_i$ on buyers corresponding to $B^L$. We can once again use the flow partition lemma to this flow with $B^H$ (or rather the elements in $B_1$ corresponding to $B^H$) as the desired set. Notice that the min-cost flow with only the buyers from $B^H$ having flow of $x'_{i_1} = x^b_i$ must indeed coincide with $\vec{z^b}(B^H)$. Therefore, we can further sub-divide the profit as,
%
%

\begin{align*}
\pi \geq & \left(\sum_{i \in B^H} \bar{p}_i x^b_i - C(\vec{z^b}(B^H))\right) + \left(\sum_{i \in B^L} \bar{p}_i \alpha x^b_i - (C(\vec{z}(B_1)) - C(\vec{z^b}(B^H)))\right)\\
\geq & \left(\sum_{i \in B^H} \bar{p}_i x^b_i - C(\vec{z^b}(B^H))\right) + \left(\alpha\sum_{i \in B^L} \bar{p}_i x^b_i - (C(\vec{z}(B_1)) - C(\vec{z^b}(B^H)))\right).
\end{align*}

We are almost done with the proof. The terms for $B^H$ coincide with what we are required to prove. We only need to get an $\alpha$ factor out of the terms in $B^L$. That is, all that remains is for us to show the following:

$$C(\vec{z}(B_1)) - C(\vec{z^b}(B^H)) \geq \alpha(C(\vec{z^b}) - C(\vec{z^b}(B^H))).$$

Let's recap what these allocations actually mean. 
\begin{align*}
\vec{z^b}  - & \text{All buyers sending $x^b_i$}\\
\vec{z}(B_1)  - & \text{Buyers in $B^H$ sending $x^b_i$, buyers in $B^L$ sending $\alpha x^b_i$}\\
\vec{z^b}(B^H)  - & \text{Buyers in $B^H$ sending $x^b_i$, buyers in $B^L$ sending zero}
\end{align*}

Suppose, we denote by $\vec{x^H}$ the demand vector where all buyers in $B^H$ are sending $x^b_i$ units of flow and rest $0$ and similarly by $\vec{x^L}$, the demand vector where all buyers in $B^L$ are sending $x^b_i$ units of flow and the rest are sending $0$ flow. Moreover, for any of these demand vectors, define $MC()$ be to the cost of the min-cost flow for that demand vector. What we need to show can be simplified as
$$[MC(\vec{x^H} + \alpha \vec{x^L}) - MC(\vec{x^H})] \leq \alpha [MC(\vec{x^H}+\vec{x^L}) - MC(\vec{x^H})].$$

First, it is not hard to argue that $MC$ is a convex function of the vector demand $\vec{x}$ since all items have convex production costs. Moreover, for any convex function $C$ and positive $a,x$ and $0 < \alpha\leq 1$), the following is true.

$$\frac{C(a+x) - C(a)}{a+x - a} \geq \frac{C(a+\alpha x) - C(a)}{a+\alpha x - a}.$$

Simplifying the denominators and cancelling $x$ from both of them tells us that $C(a+\alpha x) - C(a) \leq \alpha (C(a+x) - C(a))$. So, we can easily extend this property to convex functions of vector inputs, thereby showing
$$[MC(\vec{x^H} + \alpha \vec{x^L}) - MC(\vec{x^H})] \leq \alpha [MC(\vec{x^H}+\vec{x^L}) - MC(\vec{x^H})].$$ \end{proof}

\section{Some properties of MHR functions}

\begin{lemma}
\label{lem_mhrminusconst}
Let $f(x)$ be any non-increasing function of $x$ with a monotone hazard rate. Let $c(x)$ be a non-decreasing function such that the function $f(x) - c(x)$ is non-zero in the interval $[0, x)$. Then $f(x) - c(x)$ also has a monotone hazard rate in this interval.
\end{lemma}

\begin{lemma}
\label{lem_continuitymhr}
For a non-increasing monotone hazard rate inverse demand function $\lambda(x)$, the inverse $\lambda^{-1}(p)$ is uniquely defined at all prices other than $p=\lambda(0)$.
\end{lemma}

\begin{lemma}
\label{lem_subclaim_mhr}
Let $f(x)$ be a non-increasing non-negative function with a monotone hazard rate, i.e., $\frac{f(x)}{|f'(x)|}$ is non-increasing. Let $\tilde{x} > 0$ be a point satisfying $\frac{f(\tilde{x})}{|f'(\tilde{x})|} > \tilde{x}$. Then $f(0) < ef(\tilde{x})$.
\end{lemma}
Notice that since $f$ is non-increasing, $f(0)$ is the maximum value of the function.
\begin{proof}
Consider the function $g(x) = \ln f(x)$. Differentiating this gives us,
$$g'(x) = \frac{f'(x)}{f(x)}.$$
Recall that since $f(x)$ is non-increasing, its derivative cannot be positive. Integrating $g(x)$ from $\tilde{x}$ to $0$ gives us
\begin{align*}
\int_{\tilde{x}}^0 g'(x) = & \int_{\tilde{x}}^0 \frac{f'(x)}{f(x)}dx\\
g(0) - g(\tilde{x}) = & -\int_{\tilde{x}}^0 \frac{|f'(x)|}{f(x)}dx\\
\ln\frac{f(0)}{f(\tilde{x})} = &  \int_{0}^{\tilde{x}} \frac{|f'(x)|}{f(x)}dx\\
\leq & \int_0 ^{\tilde{x}} \frac{|f'(\tilde{x})|}{f(\tilde{x})}dx\\
= & \tilde{x} \frac{|f'(\tilde{x})|}{f(\tilde{x})}\\
< & \tilde{x} \frac{1}{\tilde{x}}\\
= & 1.
\end{align*}
The first inequality is due to $f$ being MHR, and the final inequality comes from the fact that $\frac{f(x)}{|f'(x)|} > \tilde{x}$. So, we have $\ln \frac{f(0)}{f(\tilde{x})} < 1$, which gives that $f(0) < ef(\tilde{x})$. \end{proof}

\begin{lemma}
\label{lem_mhrxchange}
Consider any non-increasing MHR function $f(x)$. Let $x_1 < x_2$  be two parameters such that
$$f(0) = \sqrt{e} f(x_1) \leq e f(x_2).$$
Then, $x_2 \leq 2 x_1$.
\end{lemma}

\begin{proof}
From the proof of Lemma~\ref{lem_subclaim_mhr}, we know that for any MHR function, the following is true.
\begin{equation}
\label{eqn_mhrlog}
log(\frac{f(x_0)}{f(x)}) =  \int_{x_0}^{x_1} \frac{|f'(x)|}{f(x)}dx.
\end{equation}

First, let $x_0 =0$ and $x=x_1$. We know that for all $x \in [0, x_1]$, $\frac{|f'(x)|}{f(x)} \leq  \frac{|f'(x_1)|}{f(x_1)}$.
\begin{align*}
log(\frac{f(0)}{f(x_1)})  = & \int_{0}^{x_1} \frac{|f'(x)|}{f(x)}dx \\
\implies log(\sqrt{e}) \leq & \int_{0}^{x_1} \frac{|f'(x_1)|}{f(x_1)}dx \\
\implies \frac{1}{2} \leq & x_1 \frac{|f'(x_1)|}{f(x_1)}.
\end{align*}

So, we have $x_1 \geq \frac{1}{2} \frac{f(x_1)}{|f'(x_1)|}.$ Let's go back to Equation~\ref{eqn_mhrlog} with $x_0 = x_1$ and $x = x_2$. We know that in the interval $x \in [x_1, x_2]$, $\frac{|f'(x)|}{f(x)} \geq  \frac{|f'(x_1)|}{f(x_1)}$. Applying this,

\begin{align*}
log(\frac{f(x_1)}{f(x_2)})  = & \int_{x_1}^{x_2} \frac{|f'(x)|}{f(x)}dx \\
\implies log(\sqrt{e}) \geq & \int_{x_1}^{x_2} \frac{|f'(x_1)|}{f(x_1)}dx \\
\implies \frac{1}{2} \geq & (x_2 - x_1) \frac{|f'(x_1)|}{f(x_1)}.
\end{align*}

So we have $(x_2 - x_1) \leq \frac{1}{2} \frac{f(x_1)}{|f'(x_1)|} \leq x_1$, which completes the proof that $x_2 \leq 2x_1$.
$\blacksquare$ \end{proof}

\begin{lemma}
\label{mhr_profitchanges}
Let $f(x)$ be a non-increasing monotone hazard rate function and let $x_2 > x_1$ be two points in its domain. Then, the following must be true
\begin{enumerate}
\item If $x_1$ satisfies $\frac{|f'(x_1)|}{f(x_1)} \geq \frac{1}{x_1}$, then $x_1 f(x_1) > x_2 f(x_2)$.
\item If $x_2$ satisfies $\frac{|f'(x_2)|}{f(x_2)} \leq  \frac{1}{x_2}$, then $x_2 f(x_2) > x_1 f(x_1)$.
\end{enumerate}
\end{lemma}
\begin{proof}
The starting point for both these results is what we already showed in Lemma~\ref{lem_subclaim_mhr}.  Define $g(x)$ to be $\log(f(x))$. Then using the same idea as in Lemma~\ref{lem_subclaim_mhr}, we get
$$log(\frac{f(x_1)}{f(x_2)})  =  \int_{x_1}^{x_2} \frac{|f'(x)|}{f(x)}dx.$$

\noindent\textbf{Statement 1:} \\
For all $x$ in $[x_1, x_2]$, the following is true,
$$\frac{|f'(x)|}{f(x)} \geq \frac{|f'(x_1)|}{f(x_1)} \geq \frac{1}{x_1}.$$ Therefore, we can bound the ratio between $f(x_1)$ and $f(x_2)$ as,
\begin{align*}
log(\frac{f(x_1)}{f(x_2)})  \geq &  \int_{x_1}^{x_2} \frac{|f'(x_1)|}{f(x_1)}dx\\
\geq & (x_2 - x_1) \frac{1}{x_1} \\
= & \frac{x_2}{x_1}-1.
\end{align*}
Therefore,
$$\frac{f(x_1)}{f(x_2)} \geq e^{(\frac{x_2}{x_1}-1)} > \frac{x_2}{x_1}.$$
The last inequality is strict because for any $\alpha > 1$, $e^{\alpha - 1} > \alpha$.\\

\textbf{Statement 2:} \\
For all $x$ in $[x_1, x_2]$, the following is true,
$$\frac{|f'(x)|}{f(x)} \leq \frac{|f'(x_2)|}{f(x_2)} \leq \frac{1}{x_2}.$$ Therefore, we can bound the ratio between $f(x_1)$ and $f(x_2)$ as,
\begin{align*}
log(\frac{f(x_1)}{f(x_2)})  \leq &  \int_{x_1}^{x_2} \frac{|f'(x_2)|}{f(x_2)}dx\\
\leq & (x_2 - x_1) \frac{1}{x_2} \\
= & 1 - \frac{x_1}{x_2}.
\end{align*}
Therefore,
$$\frac{f(x_1)}{f(x_2)} \leq e^{(1-\frac{x_1}{x_2})} < \frac{x_2}{x_1}.$$
The last inequality is strict because for any $0<\alpha < 1$, $\alpha e^{1-\alpha} < 1$. \end{proof}

\end{document}